\documentclass[11pt]{article}
\usepackage{fullpage}
\usepackage[utf8]{inputenc}
\usepackage[english]{babel}
\usepackage[T1]{fontenc}
\usepackage{amsmath,amssymb,amsfonts,amsthm}
\usepackage{url}
\usepackage{hyperref}
\usepackage{cleveref}
\usepackage{graphicx}
\usepackage[ruled,vlined, linesnumbered]{algorithm2e}
\usepackage{tikz}
\usetikzlibrary{cd}

\newtheorem{theorem}{Theorem}
\newtheorem*{theorem*}{Theorem}
\newtheorem{lemma}{Lemma}
\newtheorem{prop}{Proposition}

\newtheorem{cor}{Corollary}

\usepackage{authblk}

\newcommand{\C}{\mathbb{C}}
\newcommand{\Z}{\mathbb{Z}}
\newcommand{\N}{\mathbb{N}}

\newcommand{\ket}[1]{|{#1}\rangle}
\newcommand{\braket}[1]{|{#1}\rangle\langle{#1}|}

\newcommand{\Prob}{\mathbb{P}}

\newcommand{\PauliWithPhase}{{\cal P}}
\newcommand{\Pauli}{{\overline{\cal P}}}

\newcommand{\stab}{{\cal S}}

\DeclareMathOperator{\rank}{rank}
\DeclareMathOperator{\tr}{tr}
\DeclareMathOperator{\supp}{supp}
\DeclareMathOperator{\depth}{depth}
\DeclareMathOperator{\MLF}{MLF}
\DeclareMathOperator{\MLC}{MLC}

\newcommand{\GG}{{\bf G}}
\newcommand{\HH}{{\bf H}}

\newcommand{\SSS}{{\bf S}}

\newcommand{\circuit}{{\cal C}}
\newcommand{\outcome}{{\cal O}}
\newcommand{\spacetime}{{\cal Q}}

\DeclareMathOperator{\level}{level}
\DeclareMathOperator{\eff}{eff}

\newcommand{\propagation}[1]{\overrightarrow{#1}}
\newcommand{\backpropagation}[1]{\overleftarrow{#1}}

\usepackage{xcolor}

\title{Spacetime codes of Clifford circuits}

\author[]{Nicolas Delfosse, Adam Paetznick}

\affil[]{Microsoft Quantum, Redmond, Washington 98052, USA}

\begin{document}

\maketitle

\date

\begin{abstract}
We propose a scheme for detecting and correcting faults in any Clifford circuit.
The scheme is based on the observation that the set of all possible outcome bit-strings of a Clifford circuit is a linear code, which we call the ``outcome code''.
From the outcome code we construct a corresponding stabilizer code, the ``spacetime code''. Our construction extends the circuit-to-code construction of Bacon, Flammia, Harrow and Shi~\cite{bacon2015sparse}, revisited recently by Gottesman~\cite{gottesman2022spacetimecode}, to include intermediate and multi-qubit measurements.
With this correspondence, we reduce the problem of correcting faults in a circuit to the well-studied problem of correcting errors in a stabilizer code.
More precisely, a most likely error decoder for the spacetime code can be transformed into a most likely fault decoder for the circuit.
We give efficient algorithms to construct the outcome and spacetime codes.
We also identify conditions under which these codes are LDPC, and give an algorithm to generate low-weight checks, which can then be combined with efficient LDPC code decoders.
\end{abstract}

Several small quantum computing platforms are available today. However, the high noise rate of quantum hardware is a major obstacle to the scalability of these devices.
Some form of quantum error correction is likely necessary in order to achieve noise low enough for solving large-scale industrial problems.

The traditional solution to the noise problem is to execute quantum algorithms on error-corrected qubits. First, choose a quantum error correction code to encode each qubit. Second, design a syndrome extraction circuit for this code. This quantum circuit is run regularly and returns a bit string, the so-called syndrome, which is then used to identify errors.
This code-centric approach has limitations, however.
For example, Floquet codes are not obtained through this traditional solution and are instead defined directly by a circuit.

In this work, we take a circuit-centric approach. 
We describe a general method for correcting faults in Clifford circuits.
This methods applies not only to syndrome extraction circuits of stabilizer codes and Floquet codes, but also to general Clifford circuits that were not initially designed as a syndrome extraction circuit. The basic idea is that these circuits can include some other type of redundancy that can be exploited to correct circuit faults.
Such a circuit-centric approach was also considered previously in~\cite{bacon2015sparse, gottesman2022spacetimecode, gidney2021stim}.
Clifford circuits are central for quantum computing because they implement standard protocols such as quantum teleportation, preparation of Bell states or error correction with stabilizer codes.
Even though they are not universal for quantum computing directly, universality can be achieved with Clifford circuits by the injection of magic states~\cite{bravyi2005universal}.

We consider circuits made with unitary Clifford gates and Pauli measurements. We allow all Pauli measurements, not just single-qubit measurements, and we allow for internal measurements that can occur at any time step of the circuit.

Collectively, the measurement outcomes of a circuit produce a classical bit-string. Our basic idea is to correct circuit faults using redundancy in these bit-strings.
We prove that the set of possible outcome bit-strings of a Clifford circuit is a linear code (up to relabeling the measurement outcomes); see Theorem~\ref{theorem:outcome_code} and Corollary~\ref{cor:outcome_code_linearization}. We call this the ``outcome code'' of the circuit.
Moreover, we design an algorithm that returns a complete set of checks for the outcome code; see Algorithm~\ref{algo:circuit_checks}.

The ``outcome code'' can be used directly to detect and correct circuit faults. Doing so, however, requires designing a decoder, generally a non-trivial task.
Instead of designing a new decoder that maps check values onto circuit faults,
we construct a stabilizer code associated with the circuit, the ``spacetime code''; see Theorem~\ref{theorem:space_time_code}. Measuring the generators of the spacetime code provides corresponding values of the checks of the outcome code.
We show how to design a circuit decoder that returns a most likely fault configuration using a most likely error decoder for the spacetime stabilizer code; see Theorem~\ref{theorem:circuit_decoder_from_stabilizer_decoder}.

The spacetime code defined in this paper is related to the circuit-to-code construction of Bacon, Flammia, Harrow and Shi~\cite{bacon2015sparse}. 
There, the authors proposed a transformation of a restricted class of Clifford circuits with the goal of building new subsystem codes.
They considered a subclass of post-selection circuits and computed the parameters of the resulting subsystem code as 
a function of the input circuit.
Recently, Gottesman extended this formalism to circuits with Clifford unitaries and single-qubit preparation and measurement~\cite{gottesman2022spacetimecode}.\footnote{The term \emph{spacetime code} follows terminology introduced by~\cite{gottesman2022spacetimecode}.}
The spacetime code we consider can be seen as the stabilizer code associated with the subsystem code of~\cite{bacon2015sparse, gottesman2022spacetimecode} after generalizing their constructions to arbitrary Clifford circuits.
The stabilizer generators of the spacetime code we consider here
coincide with ``detectors'' considered by \cite{mcewen2023relaxing}, where they were used to optimize surface codes for particular architecture constraints.

Throughout, we provide new proofs of the properties of the spacetime code based on the relation between the forward and the backward propagation of the faults through the circuit.
Namely, we prove that backward propagation is the adjoint of forward propagation. 
This alternative approach unifies the treatment of all Clifford circuits and sheds some light on the underlying mathematical structure of the spacetime code.
We believe that this relation could also be relevant for other applications.

Equipped with the outcome code,
one can build a scheme for the correction of any small Clifford circuit using a lookup decoder.
To push the range of application further, we propose an algorithm to produce a set of low-weight stabilizer generators for the spacetime code~\ref{algo:spacetime_code_local_genertors}, resulting in a Low-Density Parity-Check (LDPC) spacetime code for which efficient decoders exist~\cite{gallager1962low, panteleev2021degenerate, roffe_decoding_2020, delfosse2022toward}.
Starting from a local code in $D$ dimensions, that is, a code defined by local stabilizer generators in a $D$-dimensional grid of qubits, our algorithm produces local stabilizer generators in $D+1$ dimensions for which topological decoders such as the Renormalization Group decoder~\cite{bravyi2013quantum} can be used.

Our fault correction scheme applies to any Clifford circuit, but syndrome extraction circuits are of special interest.
The design of a quantum error correction scheme is a non-trivial task which requires
(i) a syndrome extraction circuit,
(ii) a syndrome map,
(iii) a decoder.
With our work, the construction of the syndrome map and the decoder can be automated in some cases.
We do not provide a performance guarantee for the resulting scheme, and we believe that some specialized schemes, highly optimized for a specific circuit, are likely to perform better.
The main advantage of our approach is its flexibility.
Our approach only requires the circuit to be given as an input and it applies to codes implemented with Clifford operations or Pauli measurements. This includes for instance CNOT-based surface codes~\cite{dennis2002topological, fowler2012surface} or color codes~\cite{bombin2006topological}, Majorana-based surface codes~\cite{chao2020optimization} or Floquet codes~\cite{hastings2021dynamically} that are implemented with only joint measurements.

This article is organized as follows.
In order to motivate our results, we start by describing an application of the outcome code and the spacetime code in Section~\ref{sec:applications}. Namely, we discuss the implementation of an automated scheme for the correction of faults in Clifford circuits. The rest of the paper provides a complete description of all the ingredients.
Background material is introduced in Section~\ref{sec:background}.
Then, Section~\ref{sec:propagation} proves some core technical results and establishes the relation between the propagation and the backpropagation operators.
The outcome code is defined in Section~\ref{sec:outcome_code} which also describes an algorithm (Algorithm~\ref{algo:circuit_checks}) to compute a complete set of checks for this code.
The spacetime code is introduced in Section~\ref{sec:space_time_code} where we prove that a most-likely error decoder for the spacetime code can be converted into a circuit decoder that returns a most likely set of faults (Theorem~\ref{theorem:circuit_decoder_from_stabilizer_decoder}).
Section~\ref{sec:LDPC} provides an algorithm (Algorithm~\ref{algo:spacetime_code_local_genertors}) to generate low-weight generators for the spacetime code.

\begin{figure}
\centering
\includegraphics[scale=.5]{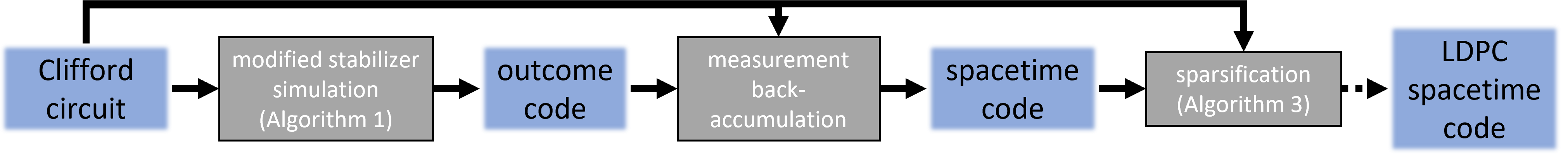}

\caption{
Construction of codes from a Clifford circuit.
Given a Clifford circuit as input, a modified stabilizer simulation, Algorithm~\ref{algo:circuit_checks}, produces the outcome code.
The corresponding spacetime code can then be constructed by accumulating measurement observables from the outcome code backward through the circuit. See Section~\ref{subsec:check_operators}.
Low-weight generators of the spacetime code are obtained by sparsification with Algorithm~\ref{algo:spacetime_code_local_genertors}.  The dashed arrow indicates that sparsification is not possible for all circuits.
All constructions run in polynomial time.
The LDPC spacetime code can be decoded using any quantum LDPC code decoder, 
thereby providing an automated and efficient means of correcting faults given only the circuit as an input.
}
\label{fig:diagram}
\end{figure}

\section{Applications of the outcome code and the spacetime code}
\label{sec:applications}

In this section, we combine all the ingredients developed in this article to describe a flexible scheme for the correction of faults in Clifford circuits. The full protocol is illustrated in Figure~\ref{fig:diagram}.
Detailed descriptions of all the ingredients are provided in subsequent sections.

\subsection{Standard design procedure for a quantum error correction scheme}

To emphasize the advantage of our approach, we first review the general approach to designing and simulating a quantum error correction scheme. For simplicity, we focus on stabilizer codes and consider the standard circuit noise model explained in Section~\ref{sec:background}.

To describe a complete quantum error correction scheme for a given stabilizer code, we must design the following components:
\begin{enumerate}
    \item {\bf Syndrome extraction circuit:} A quantum circuit that takes as input a noisy encoded state and returns a bit string.
    \item {\bf Syndrome map:} A classical procedure that takes as input the outcome of the syndrome extraction circuit and returns bit string that we call the syndrome.
    \item {\bf Decoder:} A classical procedure that takes as input the syndrome and that returns a correction to apply to the encoded state.
\end{enumerate}

These three components are non-trivial and tightly coupled, in general. Designing a syndrome extraction circuit requires working within the constraints of a specific architecture. 
See, for example,~\cite{chao2020optimization, gidney2022pair, reagor2022hardware}. 
The decoding problem is notoriously difficult~\cite{berlekamp1978inherent, iyer2015hardness}.
Some specific cases have been worked out.
For example, with the surface code, one can consider the standard syndrome extraction circuit~\cite{fowler2012surface}, a carefully chosen sequence of CNOT gates and measurements of ancillary qubits.
The syndrome of this circuit is obtained by xor of consecutive outcomes for each ancilla.
The syndrome is then input to a surface code decoder such as a Minimum Weight Perfect Matching decoder~\cite{dennis2002topological} or Union-Find decoder~\cite{delfosse2021almost}.

The effectiveness of the decoder depends tightly on the choice of syndrome map.  Indeed, the surface code decoders mentioned above may be ineffective for other choices of syndromes.  The syndrome map, in turn, depends on the details of both the code and the circuit; see Floquet codes~\cite{hastings2021dynamically}. Any change to the circuit or the code requires corresponding modification of the syndrome map.

\subsection{Automated circuit-fault correction}

Here, we propose a protocol described in Figure~\ref{fig:diagram} for correcting faults that requires only the Clifford circuit as an input, instead of the three components required in the previous section.
Consider a Clifford circuit $\circuit$ with $m$ measurements.
The first step of the protocol is the computation of a set of outcome checks $u_1, \dots, u_r \in \Z_2^m$ using Algorithm~\ref{algo:circuit_checks}.
In the absence of faults, running the circuit $\circuit$ produces an outcome bit-string $o \in \Z_2^m$ that satisfies $(o | u_1) = \dots = (o | u_r) = 0$.
Therefore, one can use these values to detect and correct circuit faults.
Moreover, the set of outcome checks returned by Algorithm~\ref{algo:circuit_checks} is maximal. It fully describes the set of possible outcomes for the circuit, which means that we are using all the information available to detect and correct faults.

Given the outcome checks, we need to design a syndrome map and a decoder in order to transform the check values $(o | u_i)$ into a correction to apply to the output state of the circuit.
For small circuits, this can be done with the syndrome map $o \mapsto \sigma$ where $\sigma_i = (o|u_i)$ and a lookup decoder.
Naively, we can construct a lookup decoder for the circuit $\circuit$ using the outcome checks $u_i$ as follows.
Loop over all the fault configurations of the circuit with up to $M$ faults for some integer $M$. For each fault configuration $F$, compute the syndrome value $\sigma$ and store a configuration made with a minimum number of faults for each syndrome $\sigma$.

Lookup decoders are impractical for large codes.
To make the decoding of large circuit possible, we must restrict the set of schemes we consider because the decoding problem for linear codes and stabilizer codes is generally intractable~\cite{berlekamp1978inherent, iyer2015hardness}.
We focus on circuits that admit local redundancy in the sense that the spacetime code of the circuit has many low-weight stabilizers.
This assumption is even more justified because the value of an outcome check corresponding to a high-weight stabilizer in the spacetime code will be very noisy and unreliable.

After computing the outcome checks $u_1, \dots, u_r$ using Algorithm~\ref{algo:circuit_checks}, we compute the corresponding stabilizers $\backpropagation{F(u_1)}, \dots, \backpropagation{F(u_r)}$ defined in Eq.~\eqref{def:relation_operator}.
Then, we run Algorithm~\ref{algo:spacetime_code_local_genertors} to produce a set low-weight stabilizer generators for the spacetime code.
These generators define a new syndrome map whose components are related to the check values $(o | u_i)$ by a linear map.
We are now equipped with a set of low-weight generators and therefore one can use any LDPC code decoder for this code.
Existing LDPC decoder options include the Union-Find decoder for local topological codes~\cite{delfosse2021almost} or for LDPC codes~\cite{delfosse2022toward}.
Other decoding strategies can be used, such as the Renormalization Group decoder~\cite{bravyi2013quantum} for topological codes or a Belief Propagation decoder for LDPC codes~\cite{panteleev2021degenerate, roffe_decoding_2020}.

\section{Background}
\label{sec:background}

\subsection{Linear codes}

A {\em linear code} with length $n$ is a $\Z_2$-linear subspace of $\Z_2^n$.
It encodes $k$ bits into $n$ bits where $k$ is the dimension of the code. 
The {\em dual code} of a linear code $C$ with length $n$, denoted $C^\perp$, is the set of vectors of $\Z_2^n$ that are orthogonal with all the vectors of $C$.
Orthogonality is with respect to the binary inner product
$
(u | v) = \sum_{i=1}^n u_i v_i \pmod 2.
$
%Vectors of a dual code are sometimes called ``checks'' because the code is fully defined by the relations $(u_i | v) = 0$ where $u_i$ runs over a basis for the dual code.

A linear code can be defined by providing a set of generators of the code or by providing a set of generators of the dual code.
% This is because $C = (C^\perp)^\perp$.
Given a set of generators of the dual code $u_1, \dots, u_r \in \Z_2^n$, the code is the set of vectors $v \in \Z_2^n$ that satisfy $(u_i | v) = 0$ for all $i=1, \dots, r$.
We refer to the vectors $u_i$ as the ``checks'' of the code.

The {\em syndrome map} associated with the checks $u_1, \dots, u_r \in \Z_2^n$ is the map 
\begin{align}
\sigma: \Z_2^n \longrightarrow \Z_2^r
\end{align}
that sends a vector $v$ onto the vector $s$ whose $i$th component is $s_i = (u_i | v)$.
The vector $s$ is called the {\em syndrome} of $u$.
The syndrome can be used to correct some bit flips of the corresponding vector and to map it back to the code space.

\subsection{Stabilizer codes}

A {\em stabilizer code} with length $n$ is defined by a set of commuting $n$-qubit Pauli operators $\stab$, that we call stabilizer generators, such that the group $\langle\stab\rangle$ they generate does not contain $-I$.
We refer to the group $\langle\stab\rangle$ as the {\em stabilizer group} of the code and a Pauli operator $P$ such that $\pm P \in \langle\stab\rangle$ is called a stabilizer.

The code space of a stabilizer code with length $n$ is the subspace of $(\C_2)^{\otimes n}$ that is invariant under the stabilizer generators.
If $\stab$ contains $r$ independent operators, the code space is a subspace of $(\C_2)^{\otimes n}$ isomorphic with $(\C_2)^{\otimes n-r}$.
We interpret this subspace as the encoding of $k = n-r$ (logical) qubits into $n$ (physical) qubits.

Error correction with a stabilizer code is based on the measurement of a set of stabilizer generators $\stab = \{S_1, \dots, S_r\}$.
This produces an outcome $\sigma \in \Z_2^r$ called the {\em syndrome}.
Here the outcome $\sigma_i$ corresponds to the eigenvalue $(-1)^{\sigma_i}$ for the $i$th measured stabilizer generator.
If the state of the system before syndrome measurement is a code state suffering from Pauli error $E \in \Pauli_n$, then the measurement returns a syndrome $\sigma$ such that $\sigma_i = [E, S_i]$
where $[E, S_i]$ is 0 if $E$ and $S_i$ commute and 1 otherwise.
In the absence of error, the syndrome is trivial.
Therefore, a non-trivial syndrome can be used to detect and correct errors on encoded states.

We denote by $\PauliWithPhase_n$ the set of $n$-qubit Pauli operators and by $\Pauli_n$ its quotient by the phase operators $\{\pm I, \pm i I\}$.
In other words, in $\Pauli_n$ we consider Pauli operators up to a phase.
It is natural to consider Pauli errors up to a phase because a global phase has no effect on quantum states.
Given a probability distribution over Pauli errors,
a {\em most likely error decoder} (MLE decoder) is a map $\Z_2^r \mapsto \Pauli_n$ that sends a syndrome $\sigma$ to the Pauli $E$ with maximum probability among Pauli errors with syndrome $\sigma$.

A {\em logical operator} is a Pauli operator that commutes with all stabilizer generators of the code.
A logical operator is {\em non-trivial} if it is not a stabilizer.
For any subset $A$ of $\Pauli_n$, we use the notation 
\begin{align}
A^\perp = \{ Q \in \Pauli_n \ | \ \forall P \in A, [P, Q] = 0 \}
\end{align}
for the set of $n$-qubit Pauli operators that commute with all the elements of $A$.
If $\langle\stab\rangle$ is a stabilizer group, then $\stab^\perp$ is the set of logical operators and $\stab^\perp \backslash \langle \stab \rangle$ is the set of non-trivial logical operators of the code.

\subsection{Clifford circuits}

We consider circuits composed of Clifford unitaries and Pauli measurements.
Qubit preparation of the state $\ket 0$ can be emulated by measuring $Z$.
We assume that all the qubits are present at the beginning of the circuit and no qubit can be added through the circuit. This is not a restriction because a qubit can be reinitialized using a single qubit measurement as explained above.
We do not impose any restriction on the size of the support of allowed operations, but our results also apply to circuits limited to single-qubit and two-qubit gates or contexts in which connectivity is limited to nearest neighbors, for example.
Our results also hold for Clifford circuits with Pauli fixes, but are not included here to keep the notation simple.

A {\em Clifford circuit} is a sequence of $s$ Clifford operations $\circuit = (C_1, \dots, C_s)$ applied to $n$ qubits.
Each operation $C_i$ is applied at a given time step, denoted $\level(C_i) \in \{1, 2, \dots\}$, that we call the {\em level} of the operation.
The {\em depth} of a circuit, denoted $\Delta(\circuit)$, is the maximum level of its operations.
To guarantee that they can be implemented simultaneously, we require that two operations with the same level have disjoint support.
We assume that the circuit operations are given in chronological order, that is $i \leq j$ implies $\level(C_i) \leq \level(C_j)$.

We refer to the state of the $n$ qubits of the circuit before the first circuit operation as the {\em input state} of the circuit and the final state of the $n$ qubits is the {\em output state} of the circuit.
We do not place any constraint on the input state of the circuit.

Throughout, $m$ denotes the number of Pauli measurements in the circuit and 
we denote by $S_1, \dots, S_m$ the measured operators. For all $j = 1, \dots, m$, the level of the measurement $S_j$ is denoted $\ell_j$.
By definition, we have $0 \leq m \leq s$, where $s$ is the number of operations of the circuit, and each $\ell_j$ is an integer in the set $\{1, \dots, \Delta(\circuit)\}$.

\subsection{Circuit faults}

We consider a standard circuit noise model where each circuit operation $C_i$ and each idle qubit is faulty with probability $p_i$.
If a unitary gate or an idle qubit is faulty, it is followed by a uniform random Pauli error $E$ acting on its support.
A faulty measurement is followed by a uniform random Pauli error $E$ acting on its support combined with a flip of the measurement outcome with probability 1/2.

Following \cite{bacon2015sparse}, it is convenient to represent faults as Pauli operators acting on half-integer time steps as shown in Fig.~\ref{fig:relation_code_stabilizer}.
A {\em fault operator} for a circuit $\circuit$ acting on $n$ qubits with depth $\Delta$ is a Pauli operator $F \in \Pauli_{n(\Delta+1)}$ acting on $n (\Delta+1)$ qubits indexed by pairs $(\ell + 0.5, q)$ where $\ell \in \{0, 1, \dots, \Delta\}$ represents a level of the circuit and $q \in \{1, \dots, n\}$ corresponds to a qubit.
We include $\ell = 0$ to represent faults on the input qubits of the circuit.
The component of $F$ on qubit $(\ell + 0.5, q)$ is denoted $F_{\ell + 0.5, q} \in \{I, X, Y, Z\}$.
It corresponds to the fault occurring right after level $\ell$ on qubit $q$.
We also use the notation $F_{\ell + 0.5}$ for the $n$-qubit Pauli operator $\otimes_{q =1}^n F_{\ell + 0.5, q}$ which represents the fault occurring right after the level $\ell$ of the circuit.

The flip of a measurement outcome can also be represented as a fault operator.
Let $P \in \Pauli_n$ be a measured Pauli operator and let $Q \in \Pauli_n$ be a weight-one operator $Q$ that anti-commutes with $P$.
Then, the fault operator $F$ such that $F_{\ell - 0.5} = Q$ and $F_{\ell + 0.5} = Q$ is a representation of the flip of the outcome of the measurement of $P$ at level $\ell$.

We assume a noise model that describes the probability of each combination of circuit faults. The corresponding probability distribution over the set of fault operators is denoted by $\Prob_{\cal F}$.

\subsection{Propagation of Pauli faults and cumulant}

The effect of a set of faults on the outcomes of a circuit can be determined by propagating the faults through the circuit as shown in Fig.~\ref{fig:relation_code_stabilizer}.

\begin{figure}
\centering
\includegraphics[scale=.8]{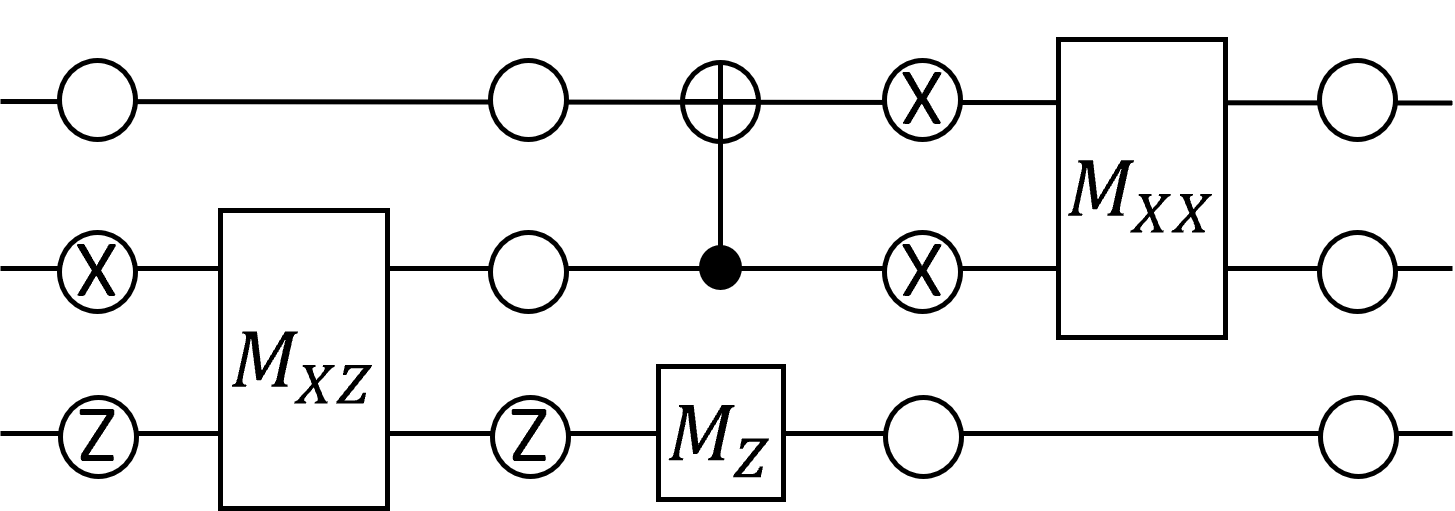}

(a)

\vspace{1cm}
\includegraphics[scale=.8]{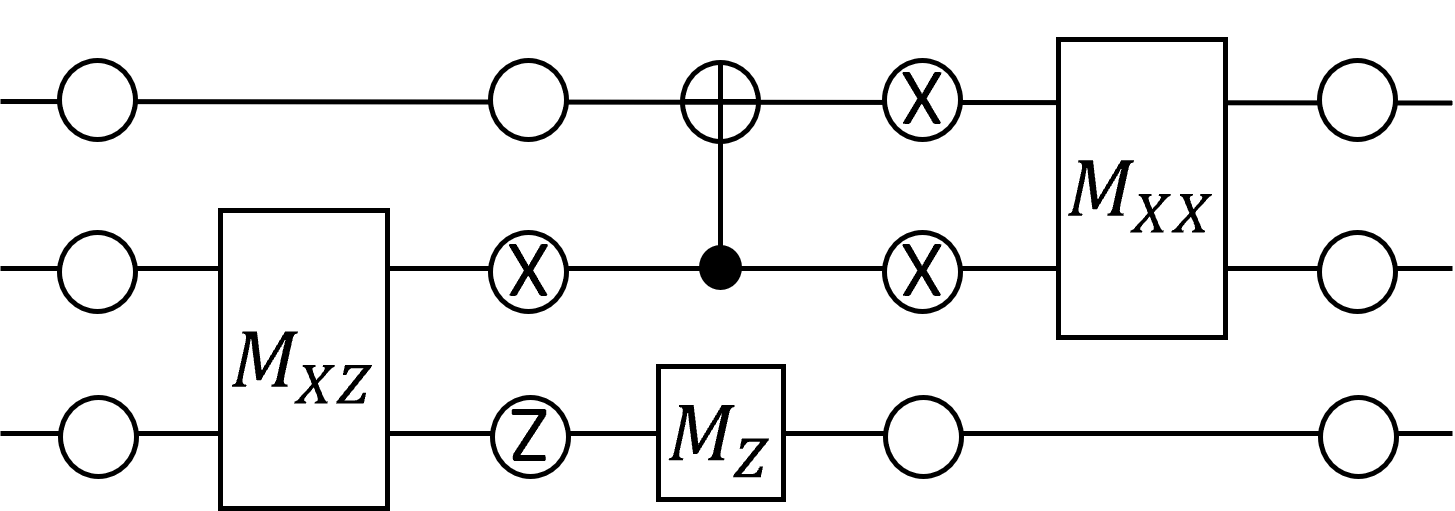}

(b)

\caption{A depth-three circuit with Pauli measurements and unitary Clifford gates.
Pauli faults are supported on the white circles.
We show a fault operator in (a) and its back-cumulant in (b) obtained by propagation of faults backward.
In this circuit, the third measurement is redundant and the three measurement outcomes $o_1, o_2, o_3$ satisfy $o_1 + o_2 + o_3 = 0 \pmod 2$. This defines a check of the outcome code.
The stabilizer generator of the spacetime code corresponding to this check is obtained by the backward accumulation of the three measurements through the circuit.
This results in the weight-four Pauli operator shown in (b) supported on the circles.
}
\label{fig:relation_code_stabilizer}
\end{figure}

The basic idea of fault propagation is that a Pauli fault $P \in \Pauli_n$ occurring before a unitary operation $U$ is equivalent to a fault $U P U^{-1}$ after the operation $g$. We refer to $U P U^{-1}$ as the {\em conjugation} of $P$ by $U$.
In the case of a Pauli measurement, a fault propagates through the corresponding projector unchanged, and it leads to a flip of the measurement outcome if the Pauli fault preceding the measurement anti-commutes with the measured operator.

The {\em cumulant} of a fault operator $F$ through a circuit $\circuit$ with depth $\Delta$, is the fault operator denoted $\propagation F$ obtained by the following procedure.
%\noindent
%{\bf Propagation procedure:}
\begin{enumerate}
\item Initialize $\propagation F$ as $\propagation F = F$.
\item For all levels $\ell = 1, 2, \dots, \Delta$ do:
\begin{enumerate}
\item [(a)] Let $E = \propagation F_{\ell - 0.5}$.
\item [(b)] Conjugate $E$ by the product of all unitary operations of $\circuit$ with level $\ell$. 
\item [(c)] Multiply $\propagation F_{\ell + 0.5}$ by $E$.
\end{enumerate}
\end{enumerate}
When propagating $\propagation F_{\ell - 0.5}$ through the operations with level $\ell$, the order in which these operations are selected is not relevant because they do not overlap.

The cumulant $\propagation F$ is defined in such a way that the operator $\propagation{F}_{\ell + 0.5}$ is the Pauli error resulting from the combination of all the faults occurring during the first $\ell$ levels of the circuit.
The cumulant was introduced in~\cite{bacon2015sparse} as the spackle operator.

\subsection{Effect of circuit faults}

Consider a depth-$\Delta$ Clifford circuit $\circuit$ acting on $n$ qubits.
Assume that the circuit contains $m$ Pauli measurements.
The $m$ measured operators are denoted by $S_1, \dots, S_m$ and for all $j = 1, \dots, m$ let $\ell_j$ denote the level of the measurement of the operator $S_j$.

Executing the quantum circuit $\circuit$ produces two types of data: 
\begin{itemize}
    \item the {\em outcome bit-string} $o \in \Z_2^m$ is the bit-string with $j$th component corresponding to the outcome of measurement $S_j$ and
    \item the {\em output state} $\rho_o$ is the state of the $n$ qubits of the system after running the circuit.
\end{itemize}
The outcome bit-string is generally not deterministic. We denote by $\Prob_{\circuit, \rho}$ the {\em outcome distribution} of the circuit when the input state of the circuit is $\rho$. Given input $\rho$, the circuit $\circuit$ yields the outcome bit-string $o$ with probability $\Prob_{\circuit, \rho}(o)$.
The output state $\rho_o$ generally depends on the outcome bit-string $o$.
A set of circuit faults represented by a fault operator $F \in \Pauli_{n(\Delta+1)}$ can affect the outcome bit-string and the output state of the circuit.

Our goal is to protect the outcome bit-string and the output state of the circuit from the effect of faults occurring throughout the circuit. 
Faults are represented by a fault operator $F \in \Pauli_{n(\Delta+1)}$.
We use the notation $\Prob_{\circuit, \rho}^{(F)}$ for the outcome distribution of the circuit $\circuit$ with input state $\rho$ and faults corresponding to $F$. The corresponding output state is denoted by $\rho_{o}^{(F)}$.

The impact of a fault $F$ on outcome $o_j$ is determined by the cumulant of $F$ just prior to measurement $S_j$. If $\propagation F_{\ell_j - 0.5}$ anti-commutes with $S_j$, then the outcome is flipped.  Likewise, the impact on the output state $\rho_o$ is determined by the cumulant of $F$ at the end of the circuit.

\begin{prop} [Effect of faults] \label{prop:fault_effects}
If $F \in \Pauli_{n(\Delta+1)}$ then, we have
\begin{itemize}
    \item $\Prob_{\circuit, \rho}^{(F)}(o) = \Prob_{\circuit, \rho}(o+f)$ where $f = (f_1, \dots, f_m) \in \Z_2^m$ such that $f_j = [\propagation{F}_{\ell_j - 0.5}, S_j]$,
    \item $\rho_{o}^{(F)} = E \rho_{o+f} E$ where $E = \propagation{F}_{\Delta + 0.5}$.
\end{itemize}
\end{prop}

The pair $(f, E)$ associated with a fault operator $F$ is called the {\em effect} of $F$ and is denoted $\eff(F)$.
We say that $F$ {\em flips the $j$-th outcome of the circuit} if $[\propagation{F}_{\ell_j - 0.5}, S_j] = 1$ and we refer to the Pauli error $E = \propagation{F}_{\Delta + 0.5}$ as the {\em residual error} of $F$.

\begin{proof}
By definition of the cumulant $\propagation{F}$, the accumulation of all the circuit faults occurring before the measurement of $S_j$ is equivalent to the Pauli error $\propagation{F}_{\ell_j - 0.5}$.
As a result, an outcome bit-string $o$ is mapped onto the vector $o+f$ given in the proposition.
When the vector $o$ is observed, the output state is the state $E \rho_{o+f} E$ because $o$ was mapped onto $o + f$ and the accumulation of all the circuit faults acts as an error $E = \propagation{F}_{\Delta + 0.5}$ on the output state of the circuit.
\end{proof}

\subsection{Correction of circuit faults with circuit decoders}

For the circuit to be correct in the presence of faults, we must guarantee that both the outcome bit-string and the output states can be recovered.
An error affecting the outcome bit-string can be just as harmful as an error on the output state.
For example, one can prepare a state $\ket +$ by measuring $X$, which produces either $\ket +$ with outcome 0 or $\ket -$ with outcome 1. To obtain the state $\ket +$, we apply $Z$ to the qubit when the outcome 1 is observed.
An incorrect outcome in this protocol directly translates into a qubit error.

The output of the circuit lives in a stabilizer code (which is a trivial code for some circuits) and therefore it is sufficient to correct the residual error up to a stabilizer of this code.
Formally, define the {\em output stabilizer group} of the circuit, denoted $\langle\stab_o\rangle$, to be the stabilizer group that keeps invariant all the possible output states $\rho_o$ of the circuit corresponding to the outcome $o$ in the absence of noise.
For a fixed outcome $o$, we refer to an operator of $\langle\stab_o\rangle$ as an {\em output stabilizer} and the stabilizer code associated with the stabilizer group generated by $\langle\stab_o\rangle$ as the {\em output stabilizer code} of the circuit.
The output stabilizer group of a circuit can be computed efficiently as in Algorithm~\ref{algo:circuit_checks} as originally proposed in~\cite{gottesman1998heisenberg, aaronson2004improved}.

A {\em circuit decoder} for the circuit $\circuit$ is defined to be a map 
\begin{align*}
\Z_2^m & \longrightarrow \Z_2^m \times \Pauli_n \\
o & \longmapsto (\hat f, \hat E)
\end{align*}
that takes as an input an outcome bit-string $o$ and that returns an estimation of the outcome flips $\hat f$ and an estimation of the residual error $\hat E$.

Assume that the faults corresponding to $F \in \Pauli_{n(\Delta+1)}$ occur, and let $(f, E) = \eff(F)$ be the effect of $F$.
The decoder is said to be {\em successful} if $\hat f = f$ and if the error $\hat E E$ is correctable for the output stabilizer code.
Here by correctable, we mean correctable for the output stabilizer code equipped with a decoder and a noiseless syndrome extraction circuit for the measurement of the stabilizer generators.
This notion of correctability depends on the decoder used for the output code. Unless otherwise stated, we pick a minimum-weight error decoder for the output code, that is a decoder returning a minimum weight Pauli error for a given syndrome.

Denote by $\rho_M$ the density matrix of the maximally mixed state $\rho_M = \left(\frac{I}{2}\right)^{\otimes n}$.
A {\em most likely fault operator} given an outcome bit-string $o \in \Z_2^m$ is defined to be a fault operator $\hat F \in \Pauli_{n(\Delta+1)}$ that maximizes the product
\begin{align} \label{eqdef:q_MLF}
Q_{\MLF}(\hat F, o) = 
\Prob_{\circuit, \rho_M}^{(\hat F)}(o) \Prob_{\cal F}(\hat F) \cdot
\end{align}
This definition is motivated by Bayes' theorem, from which the probability of a fault operator $\hat F$ given an outcome $o$ can be written as
\begin{align} \label{eq:bayes_most_likely_fault}
\Prob(\hat F | o) 
% = \frac{\Prob(o | \hat F) \Prob(\hat F)}{\Prob(o)}
= \frac
{\Prob_{\circuit, \rho_M}^{(\hat F)}(o) \Prob_{\cal F}(\hat F)}
{\sum_{\hat F \in \Pauli_{n(\Delta+1)}} \Prob_{\cal F}(\hat F) \Prob_{\circuit, \rho_M}^{(\hat F)}(o)}
\end{align}
where the input state of the circuit is chosen to be the maximally mixed state $\rho_M$ because we do not want to assume a specific input state.
For a fixed circuit and a fixed outcome bit-string, this number is proportional with $\Prob_{\circuit, \rho_M}^{(\hat F)}(o) \Prob_{\cal F}(\hat F)$ because the denominator is independent of $\hat F$.

A circuit decoder that returns the effect of a most likely fault operator is said to be a {\em most likely fault decoder} or MLF decoder.

\section{Properties of the cumulant}
\label{sec:propagation}

This section introduces the main technical tools of this paper. In particular we introduce the back-cumulant $\backpropagation{F}$ obtained by propagating a fault operator backward, and we show in Proposition~\ref{prop:propagated_commutator} that the {\em accumulator} 
\begin{align} \label{eq:def_backcumulant_map}
\Pauli_{n(\Delta+1)} & \longrightarrow F \in \Pauli_{n(\Delta+1)} \\
F & \longmapsto \propagation{F}    
\end{align}
is the adjoint of the {\em back-accumulator}
\begin{align} \label{eq:def_cumulant_map}
\Pauli_{n(\Delta+1)} & \longrightarrow F \in \Pauli_{n(\Delta+1)} \\
F & \longmapsto \backpropagation{F} \cdot
\end{align}
This relation allows us to replace the fault propagation by the backpropagation of the measured operators and leads to the definition of stabilizer generators of the spacetime code in Section~\ref{subsec:check_operators}.

\subsection{Definition of the back-cumulant}

The {\em back-cumulant} of a fault operator $F$, denoted $\backpropagation F$, is defined similarly to the cumulant with the following procedure.
\begin{enumerate}
\item Initialize $\backpropagation F$ as $\backpropagation F = F$.
\item For all levels $\ell = \Delta, \Delta-1, \dots, 1$ do:
\begin{enumerate}
\item [(a)] Let $E = \backpropagation F_{\ell + 0.5}$.
\item [(b)] Conjugate $E$ by the inverse of the product of all unitary operations of $\circuit$ with level $\ell$. 
\item [(c)] Multiply $\backpropagation F$ by $E$.
\end{enumerate}
\end{enumerate}

Similar to the cumulant, the back-cumulant $\backpropagation F$ is defined in such a way that the operator $F_{\ell - 0.5}$ is equivalent to the accumulation (backward in time) of all the faults occurring during the levels $\ell' \geq \ell$ of the circuit.

\subsection{Explicit cumulant and back-cumulant}

The following proposition provides an explicit description of the cumulant and back-cumulant of a fault operator.

\begin{prop} [Explicit cumulant and back-cumulant] \label{prop:explicit_propagation}
Let $\circuit$ be a Clifford circuit.
Let $U_\ell$ be the product of all the unitary operations of the circuit with level $\ell$ and let $U_{i, j} = U_{j} U_{j-1} \dots U_{i+1}$ for $j \geq i$.
If $F$ is a fault operator then for all $\ell \in \{0, 1, \dots, \Delta \}$, we have
\begin{align} \label{eq:explicit_propagation}
\propagation F_{\ell + 0.5} = \prod_{i = 0}^{\ell} U_{i, \ell} F_{i + 0.5} U_{i, \ell}^{-1}
\end{align}
and
\begin{align} \label{eq:explicit_backpropagation}
\backpropagation F_{\ell + 0.5} =  \prod_{j = \ell}^{\Delta} U_{\ell, j}^{-1} F_{j + 0.5} U_{\ell, j} \cdot
\end{align}
\end{prop}
Here, we use the convention $U_{i, j} = I$ if $j \leq i$.
When $j > i$, the operation $U_{i, j}$ is defined in such a way that it maps the faults occurring right after level $i$ onto equivalent faults occurring right after level $j$.

\begin{proof}
The component $\propagation F_{\ell + 0.5}$ of $\propagation F$ is obtained by conjugating all the components $F_{i + 0.5}$ of $F$ with $i \leq \ell$ through the level $i+1, i+2, \dots \ell$ of the circuit. This corresponds to the conjugation by $U_{i, \ell}$, justifying the propagation formula.

To back-propagate $F_{j + 0.5}$ from level $j + 0.5$ to level $\ell + 0.5$ with $j \geq \ell$, we conjugate this operator by the inverse of the unitary operations of the circuit at levels $\ell, \ell-1, \dots, j+1$ in that order.
This is equivalent to the conjugation by 
$$
U_{j+1}^{-1} \dots U_{\ell-1}^{-1} U_{\ell}^{-1}
$$
which is equal to $U_{\ell, j}^{-1}$.
\end{proof}

The operators $U_{i, j}$ obey the relations
\begin{align} \label{eq:U_composition}
U_{a, c} = U_{b, c} U_{a, b},
\end{align}
\begin{align}
U_{b, c}^{-1} U_{a, c} = U_{a, b}
\end{align}
and
\begin{align} \label{eq:U_U_inverse_composition}
U_{a, b} U_{a, c}^{-1} = U_{b, c}^{-1}
\end{align}
where $a,b,c$ are three integers such that $0 \leq a \leq b \leq c \leq \Delta$.

% The first relation is immediate by definition.
% To prove the second relation we write
% \begin{align*}
% U_{a, b} U_{a, c}^{-1}
%     & = U_{b} U_{b-1} \dots U_{a+1} (U_{c} U_{c-1} \dots U_{a+1})^{-1}
% \\
%     & = U_{b} U_{b-1} \dots U_{a+1} U_{a+1}^{-1} \dots U_{c-1}^{-1} U_{c})^{-1}
% \\
%     & = U_{b+1}^{-1} \dots U_{c-1})^{-1} U_{c}^{-1}
%     = U_{b, c}^{-1} \cdot
% \end{align*}

% The proof of the third relation relation goes as follows.
% \begin{align*}
% U_{b, c}^{-1} U_{a, c}
%     & = (U_{c} U_{c-1} \dots U_{b+1})^{-1} U_{c} U_{c-1} \dots U_{a+1}
% \\
%     & = U_{b+1}^{-1} \dots U_{c-1}^{-1} U_c^{-1} U_{c} U_{c-1} \dots U_{a+1}
% \\
%     & = U_{b} \dots U_{a+1}
%     = U_{a, b} \cdot
% \end{align*}

\begin{cor} \label{cor:propagation_automorphism}
The accumulator and the back-accumulator are automorphisms of the group $\Pauli_{n(\Delta+1)}$.
\end{cor}

Thanks to this corollary, one can write 
\begin{align} \label{eq:propagation_product}
\propagation{FG} = \propagation{F}\propagation{G}
\end{align}
and
\begin{align} \label{eq:backpropagation_product}
\backpropagation{FG} = \backpropagation{F}\backpropagation{G}
\end{align}
for all fault operators $F, G$.

\begin{proof}
Based on Eq.~\eqref{eq:explicit_propagation}, we see that 
\begin{align*}
\propagation{FG}_{\ell + 0.5} 
	& = \prod_{i = 0}^{\ell} U_{i, \ell} F_{i + 0.5}G_{i + 0.5} U_{i, \ell}^{-1}
\\
	& = \prod_{i = 0}^{\ell} U_{i, \ell} F_{i + 0.5} U_{i, \ell}^{-1} U_{i, \ell} G_{i + 0.5}G_{i + 0.5} U_{i, \ell}^{-1}
\\
	& = \prod_{i = 0}^{\ell} U_{i, \ell} F_{i + 0.5} U_{i, \ell}^{-1} 
		\prod_{j = 0}^{\ell} U_{j, \ell} G_{j + 0.5} U_{j, \ell}^{-1}
\\
	& = \propagation{F}_{\ell + 0.5} \propagation{G}_{\ell + 0.5} 
\end{align*}
Therein, we can reorder the Pauli operators $U_{i, \ell} F_{i + 0.5} U_{i, \ell}^{-1}$ and $U_{j, \ell} G_{j + 0.5} U_{j, \ell}^{-1}$ because we consider Pauli operators up to a phase.
This proves that the accumulator is a group morphism.

Therefore, the accumulator is a $\Z_2$-linear map over $\Pauli_{n(\Delta+1)} \simeq \Z_2^{2n(\Delta+1)}$. To show that it is bijective, it suffices to see that its matrix is lower triangular with the identity on the diagonal in the basis obtained by concatenating Pauli operators acting on level $0, 1, \dots, \Delta$.

One can show that the back-accumulator is a group automorphism with the same technique.
\end{proof}

\subsection{Basic properties of Pauli commutators}

For two Pauli operators $P$ and $Q$, denote by $[P, Q] \in \Z_2$ the commutator of $P$ and $Q$ which takes the value $0$ if $P$ and $Q$ commute and takes the value 1 otherwise.
We will use the following basic properties of the commutator.
Let $n, m$ be integers.
For all Pauli operators $P, Q \in \Pauli_n$, we have
\begin{align} \label{eq:commutator_symmetry}
[P, Q] = [Q, P].
\end{align}
For all Pauli operators $P, Q, R \in \Pauli_n$, we have
\begin{align} \label{eq:commutator_bilinearity}
[P, QR] = [P, Q] + [P, R] \pmod 2.
\end{align}
For all Pauli operators $P, Q \in \Pauli_n$ and $P', Q' \in \Pauli_m$, we have
\begin{align} \label{eq:commutator_tensor_product}
[P \otimes P', Q \otimes Q'] = [P, Q] + [P', Q'] \pmod 2.
\end{align}
For all Pauli operators $P, Q \in \Pauli_n$ and for all Clifford unitary operations $U$, we have 
\begin{align} \label{eq:commutator_unitary}
[P, Q] = [U P U^{-1}, U Q U^{-1}].
\end{align}

\subsection{Interplay between accumulation and commutation}

Based on Proposition~\ref{prop:fault_effects}, to clarify the effect of faults on the outcomes of a circuit, it is important to understand the interplay between fault propagation and commutation; see definition of $f_j$.
The following proposition is a key technical result and is used many times in the rest of the paper.
It proves that the back-accumulator is the adjoint of the accumulator.

\begin{prop} [Adjoint of the accumulator] \label{prop:propagated_commutator}
For all fault operators $F, G$ of a circuit $\circuit$, we have
$$
[\propagation F, G] = [F, \backpropagation G] \cdot
$$
\end{prop}

\begin{proof}
Because $\propagation F = \otimes_{\ell=0}^\Delta \propagation F_{\ell + 0.5}$ and 
$G = \otimes_{\ell=0}^\Delta G_{\ell + 0.5}$, we can write 
\begin{align*}
[\propagation F, G]
	& = \sum_{\ell = 0}^{\Delta} \left[\propagation F_{\ell + 0.5}, G_{\ell + 0.5}\right]
	& \text{Using Eq.~\eqref{eq:commutator_tensor_product}}
\\
	& = \sum_{\ell = 0}^{\Delta} \left[\prod_{i = 0}^{\ell} \left(U_{i, \ell} F_{i + 0.5} U_{i, \ell}^{-1} \right), G_{\ell + 0.5}\right]
	& \text{Using Eq.~\eqref{eq:explicit_propagation}}
\\
	& = \sum_{\ell = 0}^{\Delta} \sum_{i = 0}^{\ell} \left[U_{i, \ell} F_{i + 0.5} U_{i, \ell}^{-1}, G_{\ell + 0.5}\right]
	& \text{Using Eq.~\eqref{eq:commutator_bilinearity}}
\\
	& = \sum_{\ell = 0}^{\Delta} \sum_{i = 0}^{\ell} \left[F_{i + 0.5}, U_{i, \ell}^{-1} G_{\ell + 0.5} U_{i, \ell}\right]
	& \text{Using Eq.~\eqref{eq:commutator_unitary}}
\\
	& = \sum_{i = 0}^{\Delta} \sum_{\ell = i}^{\Delta} \left[F_{i + 0.5}, U_{i, \ell}^{-1} G_{\ell + 0.5} U_{i, \ell}\right]
	& \text{Interchanging the summation order}
\\
	& = \sum_{i = 0}^{\Delta} \left[F_{i + 0.5}, \prod_{\ell = i}^{\Delta} \left( U_{i, \ell}^{-1} G_{\ell + 0.5} U_{i, \ell} \right) \right]
	& \text{Using Eq.~\eqref{eq:commutator_bilinearity}}
\\	
	& = \sum_{i = 0}^{\Delta} \left[F_{i + 0.5}, \backpropagation G_{i + 0.5} \right]
	& \text{Using Eq.~\eqref{eq:explicit_backpropagation}}
\\
	& = [F, \backpropagation G]
	& \text{Using Eq.~\eqref{eq:commutator_tensor_product}} \cdot
\end{align*}
In the last equality we use the factorizations 
$F = \otimes_{i=0}^\Delta F_{i + 0.5}$ and 
$\backpropagation G = \otimes_{i=0}^\Delta \backpropagation G_{i + 0.5}$.
\end{proof}

\section{The outcome code of a Clifford circuit}
\label{sec:outcome_code}

In this section, we show that the set of possible outcome bit-strings of a Clifford circuit is almost linear code\footnote{By almost, we mean that it is not a linear subspace but an affine subspace of $\Z_2^m$ for some integer $m$. However, by changing the description of the circuit, one can make it a proper linear code.} and we propose an algorithm to efficiently compute a set of checks for this code. Our algorithm is derived from the standard stabilizer simulation algorithm of~\cite{gottesman1998heisenberg, aaronson2004improved}. In fact, removing steps 10 to 16 in Algorithm~\ref{algo:circuit_checks}, we recover a standard algorithm that computes the output stabilizer group of a Clifford circuit.
The outcome code can be used to detect and correct faults in a Clifford circuit.

We consider a depth-$\Delta$ Clifford circuit $\circuit$ acting on $n$ qubits containing $m$ measurements. 
We denote by $S_j$ the measured operators for $j = 1, \dots, m$ and by $\ell_j$ the level of the measurement of the operator $S_j$.

In many cases, the outcomes observed through a circuit are not independent. 
For instance, if a measurement is repeated twice in a row, we expect to obtain the same outcome twice.
The following theorem proves that the set of outcomes of a Clifford circuit is an affine subspace. It is defined by a set of affine checks of the form $\sum_{k \in K} o_k = 0$ or 1.

\begin{theorem} [Outcome code of a Clifford circuit] \label{theorem:outcome_code}
The set of all possible outcome bit-strings for a Clifford circuit $\circuit$ with $m$ Pauli measurements is an affine subspace $\outcome(\circuit)$ of $\Z_2^m$. Moreover a set of affine equations defining $\outcome(\circuit)$ is computed by Algorithm~\ref{algo:circuit_checks}.
\end{theorem}

\begin{algorithm}[H]
\DontPrintSemicolon
\SetKwInOut{Input}{input}\SetKwInOut{Output}{output}
\Input{A Clifford circuit $\circuit = (C_1, \dots, C_s)$ acting on $n$ qubits.}
\Output{A set $R$ of checks for the circuit $\circuit$.}

\BlankLine
    Set $U = I$.\;
	\For{all $i=1, \dots, s$}
	{
		\If{$C_i$ is a unitary gate}
		{
		    Replace $U$ by $C_i U$.\;
		}
		\If{$C_i$ is a measurement of a Pauli operator $S_j \in \PauliWithPhase_n$.}
		{
		    Compute $S_i' = U^{-1} S_i U$.\;
		}
	}
		
	Set $R = \{\}$ and $\stab = \{\}$.\;
	\For{all $i=1, \dots, s$}
	{
    	\If{$C_i$ is a measurement of a Pauli operator $S_j \in \PauliWithPhase_n$.}
		{
			\If{$\pm S_j' \in \langle \stab \rangle$}
			{
				Decompose $\pm S_j'$ as a product of operators in $\stab$.\;
				Rewrite this decomposition in the form $\pm S_j' = \prod_{k \in K_j} S_k'$.\;
				\If{$+S_j' \in \langle \stab \rangle$}
				{
					Add the affine check $o_j + \sum_{k \in K_j} o_k = 0$ to $R$.\;
				}
				\If{$-S_j' \in \langle \stab \rangle$}
				{
					Add the affine check $o_j + \sum_{k \in K_j} o_k = 1$ to $R$.\;
				}
			}
			\Else
			{
				\If{there exists $T \in \stab$ that anti-commutes with $S_j'$}
				{
					\For{$T' \in \stab$ such that $T' \neq T$}
					{
						\If{$T'$ anti-commutes with $S_j'$}
						{
							Replace $T'$ by $TT'$ in $\stab$.\;
						}
					}
					Remove $T$ from $\stab$.\;
				}
				Add $S_j'$ to $\stab$.\;
			}
		}
	}
	\Return $R$.
	\caption{Outcome code of a Clifford circuit.}
	\label{algo:circuit_checks}
\end{algorithm}

The goal of the first subroutine of Algorithm~\ref{algo:circuit_checks} (steps $1-6$) is to move the Clifford unitaries to the end of the circuit. This is done by conjugating the measured operators $S_j$ by the previous unitary gates, producing $S_j'$.
Then, the relations between the measurement outcomes are extracted by tracking the stabilizer group of this sequence of measurements.

Except for steps $10-16$, Algorithm~\ref{algo:circuit_checks} matches the standard algorithm for stabilizer simulation~\cite{gottesman1998heisenberg, aaronson2004improved}.  The only difference is the action taken when a measured Pauli operator $S_j'$ is already an element of the stabilizer group $\langle \stab \rangle$.  In that case, the usual action is to simply output the measurement result $\pm 1$.
Instead, we determine the decomposition of $S_j'$ as a product of existing elements of $\stab$.
The formula $\pm S_j' = \prod_{k \in K_j} S_k'$ asserts that the outcome of measument $C_j$ \emph{must} match sum of outcomes from prior measurements $C_k$ the circuit.  All together, they form a check of the outcome code.

The decomposition (step 12) can be computed by a variation of Gauss-Jordan elimination in which the elements of $\stab$ are augmented to indicate the indexes of their corresponding measurements.
For a circuit with $n$ qubits and $m$ measurements, the worst-case computational complexity of this step is $O(n^2(n+m))$.
Note also that this step can be combined with checking $\pm S_j' \in \langle \stab \rangle$ (step 10), which also requires Gaussian elimination.

The proof of Theorem~\ref{theorem:outcome_code} is derived from the following lemma which is similar to Gottesman's update rules~\cite{gottesman1998heisenberg}.
\begin{lemma} [Stabilizer update rules] \label{lemma:stabilizer_update_rules}
Let $\ket \psi$ be a state of a stabilizer code $C(\stab)$ with length $n$ and assume that the Pauli operator $M \in \PauliWithPhase_n$ is measured with outcome $o \in \{-1, +1\}$.
Then, the following holds.
\begin{enumerate}
\item If $\pm M \in \langle\stab\rangle$, the outcome of the measurement of $M$ is $\pm 1$ and the state of the system is unchanged after measurement.
\item If $\pm M \notin \langle\stab\rangle$ and $M$ commutes with all the elements of $\stab$. Denote $\stab' = \stab \cup \{M\}$.
	\begin{itemize}
	\item For all outcomes $o \in \{+1, -1\}$ there exists a state of $C(\stab)$ such that the measurement of $M$ has outcome $o$ with non-zero probability.
	\item For all states $\ket {\psi'}$ in $C(\stab')$ there exists a state $\ket \psi \in C(\stab)$ such that the measurement of $M$ projects $\ket \psi$ onto $\ket {\psi'}$ with non-zero probability.
	\end{itemize}
\item If $\pm M \notin \langle\stab\rangle$ and $M$ anti-commutes with an element $S_1$ of $\langle\stab\rangle$. Let $S_1, S_2, \dots, S_r$ be a generating set for $\langle\stab\rangle$ such that $S_2, \dots, S_r$ commute with $M$. Denote $\stab' = \{ M, S_2, \dots, S_r \}$.
	\begin{itemize}
	\item For all outcomes $o \in \{+1, -1\}$ there exists a state of $C(\stab)$ such that the measurement of $M$ has outcome $o$ with non-zero probability.
	\item For all states $\ket {\psi'}$ in $C(\stab')$ there exists a state $\ket \psi \in C(\stab)$ such that the measurement of $M$ projects $\ket \psi$ onto $\ket {\psi'}$ with non-zero probability.
	\end{itemize}
\end{enumerate}
\end{lemma}

\begin{proof}
The first item is an immediate consequence of the postulates of quantum mechanics.

Consider the second item.
For the state $\ket{\psi_o} = \frac{I+(-1)^oM}{2} \ket \psi$ where $\ket \psi$ is a state of $C(\stab)$, the outcome of the measurement $M$ is $o$ with probability 1 proving the first part.
%$$
%\tr \left( \frac{I+(-1)^oM}{2} \braket{\psi_o} \right)
%= \tr \left(\braket{\psi_o}\right) 
%= 1
%$$
Moreover, the measurement of $M$ maps any state $\ket {\psi'}$ of $C(\stab')$ onto itself with probability 1.

It remains to prove the third item.
Consider a set of logical operators $\bar X_1, \bar Z_1, \dots, \bar X_k, \bar Z_k$ of the stabilizer code $C(\stab)$.
Without loss of generality, we can assume that all these logical operators commute with $M$. Otherwise, if a logical operator $L$ anti-commutes with $M$, replace $L$ by $LS_1$ which is another representative of the same logical operator.
Therefore, the operators $\bar X_1, \bar Z_1, \dots, \bar X_k, \bar Z_k$ also form a basis of logical operators for the code $C(\stab')$.

Consider the logical state $\ket{\bar x}$ in the stabilizer code $C(\stab)$ where $x \in \Z_2^k$. It is the state of $C(\stab)$ fixed by the operators $(-1)^{x_i}\bar Z_i$.
This state is mapped onto the logical state $\ket{\bar x}$ of $C(\stab')$ with probability $1/2$.
To see this, denote $\stab_x = \stab \cup \{(-1)^{x_1}\bar Z_1, \dots, (-1)^{x_k}\bar Z_k\}$.
Then, $\braket{\bar x} = 2^{-n} \sum_{S \in \stab_x} S$ and the measurement of $M$ produces the outcome $o$ with probability 
\begin{align*}
	\tr \left( \frac{I+(-1)^oM}{2} \braket{\bar x} \right)
	= 2^{-(n+1)} \left(
	\tr \sum_{S \in \langle\stab_x\rangle} S 
	+ (-1)^o \tr \sum_{S \in \langle\stab_x\rangle} MS
	\right) 
	 = 1/2
	\cdot
\end{align*}
Indeed, the first sum is $\sum_{S \in \langle\stab_x\rangle} S = 2^n$ because all the Pauli operators are traceless except the identity which has trace $2^n$.
The second sum is trivial because $MS = \pm I$ would imply $\pm M \notin \stab_x$ which is impossible because $M$ anti-commute with $\stab_x$.

Therefore, we demonstrated the existence of a state $\ket{\bar x}$ of $C(\stab)$ which leads both measurement outcomes $o = 0$ or $1$ with probability $1/2$. 
Moreover, its image when the outcome is trivial is the state $\ket{\bar x}$ of the code $C(\stab')$.
Therefore, we can reach any of the basis state $\ket{\bar x}$ of $C(\stab')$ from some state of $C(\stab)$ by measuring $M$.
By linearity of the projection, this proves that one can reach any state of the code $C(\stab')$.
\end{proof}

\begin{proof} [Proof of Theorem~\ref{theorem:outcome_code}]
Based on Lemma~\ref{lemma:stabilizer_update_rules}, the set of states that belong to the output of the subcircuit $(C_1, \dots, C_{i-1})$ is a stabilizer code.
Moreover, a set of stabilizer generators $\stab_{i-1}$ for this code is computed by updating the stabilizer generators of the circuit after each operation as in Algorithm~\ref{algo:circuit_checks}.

If the operator $C_i$ is the measurement of an operator $S_j$ such that neither $S_j$ nor $-S_j$ belong to the stabilizer group of the outcome of the circuit $(C_1, \dots, C_{i-1})$, then the outcome $o_j$ can be either 0 or 1.
However, if $\pm S_j$ belongs to the stabilizer group, its outcome is constrained and this defined a check of the affine code. This check is computed in Algorithm~\ref{algo:circuit_checks}.
\end{proof}

Without loss of generality, we can assume that the outcome code $\outcome(\circuit)$ of a Clifford circuit not only an affine subspace but also a proper $\Z_2$-linear code.
Indeed, replacing the measured operator $S_j$ by $-S_j$ for all the measurements such that $-S_j \in \langle \stab \rangle$ in Algorithm~\ref{algo:circuit_checks} makes all the checks linear without modifying the output of the circuit.
 This leads to the following corollary.

\begin{cor} [Linearization of the outcome code] \label{cor:outcome_code_linearization}
Let $\circuit$ be a Clifford circuit and $\circuit'$ the circuit obtained by running Algorithm~\ref{algo:circuit_checks}, replacing $S_j$ by $-S_j$ upon encountering condition $-S_j \in \langle\stab\rangle$.
Then the outcome code of $\circuit'$ is a linear code.
\end{cor}

For all remaining sections, we assume that the outcome code $\outcome(\circuit)$ is linear. 
For obvious reasons, we may refer to a vector of $\outcome(\circuit)^\perp$ as a check of the outcome code $\outcome(\circuit)$.
The outcome code can be immediately used to correct faults in a quantum circuit.
Indeed, if a fault configuration induces a flip of some of the circuit outcomes, then it can be detected using the values of the checks of the outcome code.

\section{The spacetime code of a Clifford circuit}
\label{sec:space_time_code}

As a linear code, the outcome code provides natural means of detecting and correcting errors in the classical bit-string of measurement outcomes of a circuit. 
However, we also need to correct the output qubits of the circuit.
In order to correct qubit errors, we need a way of mapping syndromes of the outcome code back on to corresponding qubit errors, i.e., a circuit decoder. One way to build this map is to construct a lookup table by enumerating combinations of faults recording their syndrome. But this approach scales exponentially with circuit size, and so is feasible only for small circuits.

Instead, we might hope to leverage existing techniques for constructing decoders for stabilizer codes.  In this setting, we imagine that Pauli errors can occur on qubits, but syndrome information is obtained perfectly.  Contrast this with the circuit setting in which errors can occur on any component, including measurement outcomes.

In this section, we introduce a stabilizer code, the spacetime code, associated with a Clifford circuit. We then show that the problem of correcting faults in the circuit reduces to the problem of correcting Pauli errors in this stabilizer code.

\subsection{Definition of check operators}
\label{subsec:check_operators}

Here, we define the check operators that are the stabilizers of the spacetime code.

It is convenient to introduce another notation from \cite{bacon2015sparse}.
The fault operator $F \in \Pauli_{n(\Delta+1)}$ obtained by placing a $n$-qubit Pauli operator $P$ at level $\ell + 0.5$ is denoted $\eta_{\ell + 0.5}(P)$.
By definition, we have $F_{\ell + 0.5} = P$ and $F_{\ell' + 0.5} = I$ for all $\ell' \neq \ell$.

For any vector $u \in \outcome(\circuit)^\perp$, define the {\em check operator} to be the back-cumulant $\backpropagation{F(u)}$ of the fault operator
\begin{align} \label{def:relation_operator}
F(u) = \prod_{j = 1}^{m} \eta_{\ell_j - 0.5}(S_j^{u_j})
\end{align}
Recall that $S_1, \dots, S_m$ are the $m$ measured operators of the circuit and $S_j$ is measured at level $\ell_j$.
In other words, $F(u)$ is the fault operator obtained by placing each measured operator $S_j$ with $u_j = 1$ at the level right before it is measured.

One may prefer to use a cumulant instead of a back-cumulant. Then, we can define the check operators as $\propagation{F'(u)}$ where
\begin{align}
F'(u) = \prod_{j = 1}^{m} \eta_{\ell_j + 0.5}(S_j^{u_j}) \cdot
\end{align}
The definition is discussed in Appendix~\ref{appendix:alternative_def_of_check_operators} where we prove that $\backpropagation{F(u)} = \propagation{F'(u)}$ for vectors $u$ of $\outcome(\circuit)^\perp$.
We prefer using $\backpropagation{F(u)}$ because it makes some proofs more straightforward.
What makes the back-accumulator natural is the fact that it is the adjoint of the accumulator (Proposition~\ref{prop:propagated_commutator}). Informally, this means that we can replace the propagation of faults by the backpropagation of measurements. We will make this claim rigorous throughout this section.

\begin{prop} \label{prop:check_operator_morphism}
The map $u \mapsto \backpropagation{F(u)}$ is an injective group morphism from $\Z_2^m$ to $\Pauli_{n(\Delta+1)}$.
\end{prop}

\begin{proof}
For all $u, v \in \Z_2^m$, the operators $F(u)$ and $F(v)$ satisfy $F(u)F(v) = F(u+v)$.
Combined with Eq.~\eqref{eq:backpropagation_product}, this leads to $\backpropagation{F(u)} \backpropagation{F(v)} = \backpropagation{F(u+v)}$ which proves that 
$u \mapsto \backpropagation{F(u)}$ is a group morphism.
It is injective as the composition of an injective map $u \mapsto {F(u)}$ and a bijection $F \mapsto \backpropagation F$ (Corollary~\ref{cor:propagation_automorphism}).
\end{proof}

We will use extensively the relation $\backpropagation{F(u+v)} = \backpropagation{F(u)}\backpropagation{F(v)}$
which is a consequence of Proposition~\ref{prop:check_operator_morphism}.

\subsection{The spacetime code}

The following theorem proves that the check operators $\backpropagation{F(u)}$ form a stabilizer group. We refer to the corresponding stabilizer code as the {\em spacetime code} of the circuit, denoted $\spacetime(\circuit)$.

\begin{theorem} [The spacetime code] \label{theorem:space_time_code}
Let $\circuit$ be a Clifford circuit with depth $\Delta$.
Then, the set of all the check operators $\backpropagation{F(u)}$ for $u \in \outcome(\circuit)^\perp$ is a stabilizer group.
Moreover, if the outcome code $\outcome(\circuit)$ of the circuit has parameters $[m, k]$, the spacetime code of the circuit
$\spacetime(\circuit)$ has parameters $[[N, K]]$ with $N=n(\Delta+1)$ and $K = N-(m-k)$.
\end{theorem}

\begin{proof}
Based on Proposition~\ref{prop:check_operator_morphism}, the set of operators of the form $\backpropagation{F(u)}$ with $u \in \outcome(\circuit)^\perp$ is a subgroup of $\PauliWithPhase_{n(\Delta+1)}$. Proposition~\ref{prop:check_operators_commutation} proves that this subgroup is commutative.
Moreover, it cannot contain $-I$ because the map $u \mapsto \backpropagation{F(u)}$ with value in $\Pauli_{n(\Delta+1)}$ is injective. Indeed, if there exists $u \neq 0$ such that $\backpropagation{F(u)} = -I$, then modulo a phase we have $\backpropagation{F(u)} = \backpropagation{F(0)}$ which contradicts the injectivity.
Another consequence of the injectivity is that the rank of the image stabilizer group is equal to the dimension of $\outcome(\circuit)^\perp$ which is $m-k$. This provides the stabilizer code parameters.
\end{proof}

% Let $u_1, \dots, u_r \in \outcome(\circuit)^\perp$ be the outcome checks returned by Algorithm~\ref{algo:circuit_checks}.
% The {\em syndrome of an outcome bit-string} $o \in \Z_2^m$ is defined to be the vector $\sigma_\outcome(o) \in \Z_2^r$ whose $i$th component is
% $
% (o | u_i).
% $
% The {\em syndrome of a fault operator} $F \in \Pauli_{n(\Delta+1)}$ is defined to be the vector $\sigma_\spacetime(F) \in \Z_2^r$ whose $i$th component is 
% $
% [F, \backpropagation{F(u_i)}].
% $

% \begin{prop}
% Let $o \in \Z_2^m$ be an outcome bit-string obtained by running the circuit $\circuit$ on an input state $\rho$ with faults represented by a fault operator $F$.
% Then, we have
% $
% \sigma_\outcome(o) = \sigma_\spacetime(F).
% $
% \end{prop}

% \begin{proof}
% This is an immediate application of Lemma~\ref{lemma:outcome_check_and_check_operator}.
% \end{proof}

% Based on this proposition, the problem of finding a most likely set of circuit faults given a outcome $o \in \Z_2^m$ can be solved by finding a most likely Pauli error for the spacetime code for the Pauli noise model corresponding to the distribution of circuit faults.
% This reduces the problem of correction of circuit faults to a standard error correction problem in stabilizer code.

\subsection{Application to the correction of circuit faults}

Here, we use the spacetime code to show that one can reduce the problem of correcting circuit faults to the correction of Pauli errors in a stabilizer code.

We consider a set of checks $u_1, \dots, u_r$ for the outcome code $\outcome(\circuit)$ of a Clifford circuit.
The syndrome map associated with these checks is denoted 
$$
\sigma_{\outcome}: \Z_2^m \rightarrow \Z_2^r.
$$
Recall that the $i$th syndrome bit of a vector $v\in\Z_2^m$ is $(u_i | v)$.
We consider the spacetime code $\spacetime(\circuit)$ with stabilizer generators $\backpropagation{F(u_1)}, \dots, \backpropagation{F(u_r)}$.
The corresponding syndrome map is denoted 
$$
\sigma_{\spacetime}: \Pauli_{n(\Delta+1)} \longrightarrow \Z_2^r
$$
and the $i$th syndrome bit of a Pauli error $F\in\Pauli_{n(\Delta+1)}$ is $[\backpropagation{F(u_i)}, F]$.

An MLE decoder for the stabilizer code $\spacetime(\circuit)$ is a map 
$$
D_\spacetime: \Z_2^r \longrightarrow \Pauli_{n(\Delta+1)}
$$
that returns a most likely Pauli error given a syndrome $\sigma_{\spacetime}(F)$.
Here, we consider Pauli errors according to the distribution $\Prob_{\cal F}$ of circuit faults.

Finally, recall that the effect of a fault operator $F$ in a circuit is obtained from a map 
$$
\eff: \Pauli_{n(\Delta+1)} \longrightarrow \Z_2^m \times \Pauli_n 
$$
that maps a fault operator onto its corresponding outcome bit-string and the output state.

\begin{theorem} [Circuit decoder from stabilizer code decoder] \label{theorem:circuit_decoder_from_stabilizer_decoder}
Let $\circuit$ be a Clifford circuit.
If $D_\spacetime$ is a MLE decoder for stabilizer code $\spacetime(\circuit)$, then 
$
\eff \circ D_\spacetime \circ \sigma_{\outcome}
$
is a MLF circuit decoder for the circuit $\circuit$.
\end{theorem}

We first prove a lemma.
Recall that $\rho_M$ denotes the $n$-qubit maximally mixed state.
The indicator function of a set $A$ is denoted $\delta_{A}$. It takes the value $\delta_{A}(x) = 1$ if $x \in A$ and 0 otherwise.

\begin{lemma} \label{lemma:outcome_distribution_uniformity}
Let $\circuit$ be a Clifford circuit and let $F$ be a fault operator with effect $\eff(F) = (f, E)$.
If $\outcome(\circuit)$ is the outcome code of $\circuit$, then we have
\begin{align}
\Prob_{\circuit, \rho_M}^{(F)}(o) = \frac{1}{2^k} \delta_{f + \outcome(\circuit)}(o)
\end{align}
where $k = \dim \outcome(\circuit)$.
\end{lemma}

\begin{proof}
Consider first the case of a trivial fault operators $F = I$.
Then, we know from Theorem~\ref{theorem:outcome_code} that the set of outcomes $o$ with non-zero probability is a subset of the outcome code $\outcome(\circuit)$.
The main difference with Theorem~\ref{theorem:outcome_code} is that we restrict the input state to be the maximally mixed state.
The unitary operations of the circuit preserve the maximally mixed states and applying a Pauli measurement to the maximally mixed state returns a uniform outcome and the post-measurement state is the maximally mixed state.
This proves that all outcome bit-strings of $\outcome(\circuit)$ can occur when the input state of the circuit is fixed to be the maximally mixed state $\rho_M$. 
Moreover, they occur with the same probability.
This proves that 
$$
\Prob_{\circuit, \rho_M}^{(I)}(o) = \frac{1}{2^k} \delta_{\outcome(\circuit)}(o) \cdot
$$
Injecting $F$ in the circuit shifts the outcome bit-string $o$ by $f$, which leads to the shifted indicator function in the lemma.
\end{proof}

\begin{proof} [Proof of Theorem~\ref{theorem:circuit_decoder_from_stabilizer_decoder}]
We need to prove that for all $o \in \Z_2^m$, the fault operator $F = D_\spacetime \circ \sigma_{\outcome}(o)$ maximizes the function 
$
Q_{\MLF}(F, o)
$
introduced in Eq.~\eqref{eqdef:q_MLF}.
First, denote $(f, E) = \eff (F)$.
Based on Lemma~\ref{lemma:outcome_distribution_uniformity}, we have 
\begin{align}
Q_{\MLF}(F, o)
= \frac{1}{2^k} \delta_{f + \outcome(\circuit)}(o) \Prob_{\cal F}(F)
\end{align}
which means that one can maximize $Q_{\MLF}(F, o)$ by selecting a fault operator $F$ such that $o \in f + \outcome(\circuit)$ that has maximum probability $\Prob_{\cal F}(F)$.
The condition $o \in f + \outcome(\circuit)$ is equivalent to 
$\sigma_{\outcome}(f) = \sigma_{\outcome}(o)$.
Finally, Lemma~\ref{lemma:outcome_check_and_check_operator}, we know that $\sigma_{\outcome}(f) = \sigma_{\spacetime}(F)$.
Overall, we showed that a MLF decoder for the circuit $\circuit$ is a circuit decoder that returns the effect of a fault operator $F$ such that $\sigma_{\spacetime}(F) = \sigma_{\outcome}(o)$ with maximum probability.
Therefore the map $\eff \circ D_\spacetime \circ \sigma_{\outcome}$ is a MLF decoder.
\end{proof}

\subsection{Beyond MLF circuit decoders}

A MLF decoder is sometimes good enough but it ignores the fact that two fault operators may have the same effect and two residual errors that differ in an output stabilizer can be considered equivalent.

In some cases one may prefer a {\em most likely coset decoder} (MLC decoder), which is defined to be a circuit decoder that takes as an input an outcome bit-string $o$ and returns a pair $(\hat f, \hat E)$ that maximizes the sum
\begin{align} \label{eqdef:q_MLC}
Q_{\MLC}(\hat f, \hat E, o) = 
\sum_{E' \in \hat E\langle\stab_o\rangle}
\sum_{\substack{\hat F \in \Pauli_{n(\Delta+1)} \\ \eff(\hat F) = (\hat f, E')}}
\Prob_{\circuit, \rho_M}^{(\hat F)}(o) \Prob_{\cal F}(\hat F) \cdot
\end{align}
Wherein, $\hat E\langle\stab_o\rangle$ denotes the coset of $\hat E$ in the output stabilizer group.
Based on Eq.~\eqref{eq:bayes_most_likely_fault}, this sum is proportional with the sum of the probabilities of all the fault operators that leads to an effect $(\hat f, E')$ where the Pauli error $E'$ is equivalent to $\hat E$ up to a stabilizer of the output code.

One can design a MLC decoder using a subsystem version of the spacetime code obtained by gauging logical operators which correspond to faults with trivial effect.
Namely, one can define a subsystem code from the union of two sets of gauge operators: 
(i) the operators $G(P, \ell)$, defined by Eq.~\eqref{eq:image-operator}, for each $\ell$ and for each $P$ that commutes with all the measurements at level $\ell$ and 
(ii) the operators $\eta_{\ell_j - 0.5}(S_j)$ where $S_j$ is an operator measured at level $\ell_j$.
Then, the center of the gauge group is the stabilizer group of the spacetime code.
Equivalently, the subsystem spacetime code can be seen as a variant of the stabilizer spacetime code obtained by gauging the logical qubits corresponding to the operators $G(P, \ell)$ such that $P$ commutes with all the measured operators at level $\ell$.
This subsystem code can be seen as a generalization of the circuit-to-code construction proposed in~\cite{bacon2015sparse}. Here we consider circuits including intermediate measurements and multi-qubit measurement whereas~\cite{bacon2015sparse} focuses on post-selection circuits.

Then, one could design a a MLC decoder from a decoder for the subsystem spacetime code that returns a most likely coset of the gauge group.
We do not expand on this strategy because designing a most likely coset for a subsystem code is generally challenging. In what follows, we focus on the strategy suggested by Theorem~\ref{theorem:circuit_decoder_from_stabilizer_decoder} and we use decoders for the stabilizer spacetime code.

\section{LDPC spacetime code}
\label{sec:LDPC}

To obtain a complete scheme for the correction of faults in a Clifford circuit, we need to provide a decoder. However, decoding a general code is quite difficult. Indeed, the maximum likelihood decoding problem is NP-hard for linear codes~\cite{berlekamp1978inherent} and $\#$P-hard for stabilizer codes~\cite{iyer2015hardness}.
However, some classes of codes such as LDPC codes, which are defined by low-weight checks, admit an efficient decoder with good performance~\cite{gallager1962low}.

Here, we consider restrictions induced by a limited connectivity in the quantum hardware implementing the circuit. This imposes constraints on the spacetime code which, in some cases, make it easier to decode.
Our basic idea is to produce a set of low-weight stabilizer generators for the spacetime code and to use an LDPC code decoder.
%In this section, we propose an algorithm that produces low-weight stabilizers for a spacetime code.

\subsection{Qubit connectivity and LDPC spacetime codes}

Because the back-cumulant can spread errors through the entire circuit, the check operators $\backpropagation{F(u)}$ are not expected to have low weight.
However, if the circuit connectivity is limited, then then back-cumulant does not spread $F(u)$ as quickly.
Below we identify sufficient conditions that ensure that the spacetime code admit low-weight generators.

A family of stabilizer codes is said to be {\em LDPC} if each stabilizer group admit generators acting on $O(1)$ qubits.
A family of stabilizer codes is said to be {\em $D$-dimensional local} if the $n$ qubits of the codes can be mapped onto the vertices of a $D$-dimensional grid $\Z^D$ in such a way that the stabilizer generators are all supported on a ball with radius $O(1)$ in the grid.
$D$-dimensional local codes are a subset of quantum LDPC codes.

The depth of an outcome bit-string $u \in \Z_2^m$, denoted $\depth(u)$ is defined as
$$
\depth(u) = 
\max \{ \ell_j \ | \ u_j = 1 \} 
- \min \{ \ell_j \ | \ u_j = 1 \} 
+ 1 \cdot
$$
Recall that $\ell_j$ is the level of the $j$ th measurement of the circuit.
The measurements indexed by $j$ satisfying $u_j = 1$ are supported in a subset of $\depth(u)$ consecutive levels of the circuit $\circuit$.

A family of circuits is said to be {\em bounded} if the outcomes codes of the circuits admit a set of checks with weight $O(1)$ and depth $O(1)$.

\begin{prop} [LDPC spacetime codes] \label{prop:LDPC_condition}
Let $(\circuit_t)_{t \in \N}$ be a family of Clifford circuits.
If the family of circuits is bounded and if all circuit operations act on $O(1)$ qubits, then the spacetime codes of the circuits are LDPC.
\end{prop}

\begin{proof}
The stabilizer group of the spacetime code is generated by the check operators $\backpropagation{F(u_i)}$ corresponding to the generators $u_1, \dots, u_r$ of $\outcome(\circuit_t)^\perp$ with bounded depth.
By Lemma~\ref{lemma:check_operator_support}, the stabilizer generators $\backpropagation{F(u_i)}$ have support on $O(1)$ levels. 
Moreover, because the circuits are bounded, the weight of $F(u)$ is bounded. Back-accumulating $F(u)$ through a circuit made with operations acting on $O(1)$ qubits leads to a bounded weight operator $\backpropagation{F(u)}$.
\end{proof}

\begin{prop} [$(D+1)$-dimensional spacetime codes] \label{prop:D_dim_local_condition}
Let $(\circuit_t)_{t \in \N}$ be a family of Clifford circuits acting on qubits placed on a $D$-dimensional grid $\Z^D$.
If the family of circuits is bounded and if all circuit operations act on qubits separated by a distance $O(1)$, then the spacetime codes of the circuits are $(D+1)$-dimensional local.
\end{prop}

\begin{proof}
By the same argument as in the proof of Proposition~\ref{prop:LDPC_condition}, the stabilizer generators $\backpropagation{F(u)}$ are obtained by back-accumulation of bounded weight operators $F(u)$.
Because the back-cumulant goes through local operations in $D$ dimensions, this results in local generators in $(D+1)$ dimensions.
\end{proof}

\subsection{Restriction of a stabilizer to a subset of qubits}

To generate low-weight generators for a spacetime code, we will first consider a simpler task.
We are given a stabilizer group $\langle\stab\rangle$ with length $n$ and a subset $A \subset \{1, \dots, n\}$ of the qubits. Our objective is to compute a generating set for the subgroup of operators $S \in \langle\stab\rangle$ whose support is included in $A$.

Let us first describe a naive solution.
Assume that we are given a set of stabilizer generators $\stab = S_1, \dots, S_r$.
It suffices to solve the linear system with the $n-|A|$ equations
$$
\prod_{i=1}^r (S_{i, j})^{\lambda_i} = I
$$
in the $r$ variables $\lambda_1, \dots \lambda_r \in \Z_2$ where $S_{i, j}\in\Pauli_1$ is the $j$th component of $S_i$ and $j$ varies over $\{1, \dots, n\} \backslash A$.
Using Gaussian elimination, the complexity of this approach is cubic in $n$.

This scaling may be too slow for large circuits. We can achieve a more favorable complexity in the case of a small subset $A$ using the following proposition.
Before describing our solution, we need to introduce some notations.

The {\em restriction of a Pauli operator} $P \in \Pauli_n$ to a subset of qubits $A$ is the operator $P_{|A}$ obtained by setting all the components of $P$ outside of $A$ to $I$.
The {\em restriction of a subgroup} $G$ of $\Pauli_n$ to the subset $A$, denoted $G_{|A}$, is the set that contains the restrictions of all the operators of $G$.
For example, the set of $n$-qubit Pauli operators supported on a subset $A$ of $\{1, \dots, n\}$ is $\Pauli_{n |A}$.

The following proposition yields a more efficient way to compute the restriction of a stabilizer group to a small subset of qubits.
Recall that we use the notation $G^\perp$ for the set of Pauli operators that commutes with the operators included in the subset $G$ of $\Pauli_n$.
For a stabilizer code $\stab$, the set $\stab^\perp$ is the set of logical operators of the code.

\begin{prop} [Restricted stabilizer subgroup] \label{prop:induced_pauli_subsgroup}
Let $\langle\stab\rangle$ be a length-$n$ stabilizer code with stabilizer generators $\stab = \{S_1,\dots, S_r\}$ and logical operators $\bar X_1, \bar Z_1, \dots, \bar X_k, \bar Z_k$. 
Let $A \subset \{1, \dots, n\}$.
A Pauli operator $P$ with support contained in $A$ is a stabilizer iff it commutes with the restricted operators $S_{1 |A}, \dots, S_{r |A}, \bar X_{1|A}, \bar Z_{1|A}, \dots, \bar X_{k|A}, \bar Z_{k|A}$.
\end{prop}

\begin{proof}
A Pauli operator $P$ belongs to the stabilizer group $\langle\stab\rangle$ iff it commutes with all the stabilizer generators and logical operators of the code, that is iff $[P, L] = 0$ for all $L \in \stab^\perp$.
The commutator can be decomposed as a sum over $A$ and its complement $A^C$,
$$
[P, L] = [P, L_{|A}] + [P, L_{|A^C}].
$$
For a Pauli operator $P$ with support contained in $A$ we have $[P, L_{|A^C}] = 0$ and therefore
$
[P, L] = [P, L_{|A}].
$
\end{proof}

Based on Proposition~\ref{prop:induced_pauli_subsgroup}, we design Algorithm~\ref{algo:induced_stabilizer_group} which returns a set of generators for the restriction of a stabilizer group to a subset of qubits.
If each qubit is acted on by $O(1)$ stabilizer generators $S_i$ and $O(1)$ logical operators, the matrix $\GG$ obtained at the end of line 2 has size $O(|A|) \times O(|A|)$ and it can be constructed in $O(|A|^2)$ bit operations.
Then, the most expensive subroutine of Algorithm~\ref{algo:induced_stabilizer_group} is the transformation of the matrix $\GG$ in reduced row echelon form which can be done in $O(|A|^3)$ bit operations using Gaussian elimination.

\begin{algorithm}[H]
\DontPrintSemicolon
\SetKwInOut{Input}{input}\SetKwInOut{Output}{output}
\Input{A set of stabilizer generators $\stab = \{S_1,\dots, S_r\}$ and  logical operators $\bar X_1, \bar Z_1, \dots, \bar X_k, \bar Z_k$ for a stabilizer code with length $n$.}
\Output{A generating set for the restriction of $\langle\stab\rangle$ to a subset $A \subset \{1, \dots, n\}$.}

\BlankLine
    Construct a Pauli matrix $\GG$ whose rows are the restricted operators of the form $S_{1|A},\dots, S_{r|A}, X_{1|A}, \bar Z_{1|A}, \dots, \bar X_{k|A}, \bar Z_{k|A}$.\;
    Remove the rows supported on the complement $A^C$ of $A$ and the trivial rows of $\GG$.\;
    Put $\GG$ in reduced row echelon form (using Gaussian elimination).\;
    Using the reduced row echelon form of $\GG$, construct a Pauli matrix $\HH$ whose rows satisfy $[\HH_i, \GG_j] = \delta_{i, j}$.\;
    Construct a Pauli matrix $\SSS'$ whose rows are the $2|A|$ single-qubit operators $X_{i}, Z_{i}$ for $i \in A$.\;
    \For{each row $\SSS_i'$ of $\SSS'$}
    {
        \For{each row $\HH_j$ of $\HH$}
        {
            \If{$[\SSS_i', \GG_j] = 1$}
            {
                Multiply $\SSS_i'$ by $\HH_j$.
            }
        }
    }
    \Return{the rows of $\SSS'$.}
\caption{Restricted Stabilizer group}
\label{algo:induced_stabilizer_group}
\end{algorithm}

\subsection{Construction of low-weight generators for the spacetime code}

To find a set of low-weight generators for a given stabilizer group, one could apply Algorithm~\ref{algo:induced_stabilizer_group} to all subsets of $w$ qubits with $w=1,2, \dots$ until there are enough generators to span the full stabilizer group.
%One could speed-up this search using information sets for Pauli groups~\cite{delfosse2022lookup} but the cost remains discouraging for general stabilizer codes.
In this section, we show that one can use some information about the structure of spacetime code to help probe the right subsets of qubits and generate low-weight stabilizer generators.

The {\em spacetime graph} of a circuit is a graph with vertex set $V = \{1, \dots, n\} \times \{0.5,1.5, \dots, \Delta - 0.5\}$ corresponding to the qubits supporting the fault operators except the last level.
The edges of the spacetime graph are constructed by looping over all the operations of the circuit and for each operation with level $\ell$ acting on qubits $q_{i_1}, \dots, q_{i_t}$, connecting together all the qubits of the form $(q, \ell - 0.5)$ or $(q, \ell' + 0.5)$.
An operation acting on $w$ qubits induces a clique with $2w$ vertices supported on level $\ell - 0.5$ and $\ell + 0.5$.
If the circuit is made with operations acting on at most $w$ qubits, the maximum degree of the spacetime graph is upper bounded by $2(2w-1)$.

A stabilizer of a stabilizer code is said to be {\em connected} if its support is connected in the spacetime graph.
In Appendix~\ref{sec:connected-components}, we prove Proposition~\ref{prop:stabilizer_component} that the restriction of a stabilizer of the spacetime code to any connected component of its support is a stabilizer.
This proves that stabilizers of the spacetime code can be decomposed as products of connected stabilizers.
Instead of running over all subsets of qubits, Algorithm~\ref{algo:spacetime_code_local_genertors} returns all the connected stabilizers of a spacetime code by running over the neighborhoods of vertices in the space time graph.

\begin{prop} \label{prop:algo_return_all_connected_stabilizers}
Let $\circuit$ be a Clifford circuit.
Then Algorithm~\ref{algo:spacetime_code_local_genertors} returns all the connected stabilizers of the space time code with weight up to $M$.
\end{prop}

\begin{algorithm}[H]
\DontPrintSemicolon
\SetKwInOut{Input}{input}\SetKwInOut{Output}{output}
\Input{A Clifford circuit $\circuit$. An integer $M$.}
\Output{A set containing all the connected stabilizers of the spacetime code of $\circuit$ with weight $\leq M$.}

\BlankLine
    Set $\stab_M = \{\}$.\;
    Construct a set of generators $\{u_1, \dots, u_r\}$ of the code $\outcome(\circuit)^\perp$ using Algorithm~\ref{algo:circuit_checks}.\;
    Construct the set $\stab$ of stabilizer generators $S_i = \backpropagation{F(u_1)}$ of the spacetime code.\;
    Construct a set of logical operators $\bar X_1, \bar Z_1, \dots, \bar X_K, \bar Z_K$ of the spacetime code.\;
    Construct the spacetime graph $G = (V, E)$ of $\circuit$.\;
    \For{all vertices $v \in V$}
    {
        Let $A$ be the set of vertices of $G$ at distance $\leq \lfloor M/2 \rfloor$ from $v$.\;
        Using Algorithm~\ref{algo:induced_stabilizer_group}, compute a set $\stab(A)$ of generators of $\langle\stab\rangle$ restricted to $A$.\;
        \For{ all vectors $F$ of $\langle\stab(A)\rangle$ with weight $|F| \leq M$}
        {
            Compute the connected components of the support of $F$.\;
            If the support of $F$ has a single connected component, add $F$ to $\stab_M$.\; 
        }
        % Compute all the vectors of $\langle\stab(A)\rangle$ with weight $\leq R$ and add them to $\stab_M$.\;
    }
    \Return{$\stab_M$.}
	\caption{Low weight stabilizers of a spacetime code}
	\label{algo:spacetime_code_local_genertors}
\end{algorithm}

In the case of a depth-$\Delta$ circuit acting on $n$ qubits with operations acting on at most $w$ qubits, Algorithm~\ref{algo:spacetime_code_local_genertors} explores $n\Delta$ subsets $A$ of qubits with size at most 
$
|A| \leq 1 + \sum_{i=1}^{\lfloor M/2 \rfloor} \delta(\delta-1)^{i-1}.
$
where $\delta = 2(2w-1)$ is the degree of the spacetime graph.

The proof of Proposition~\ref{prop:algo_return_all_connected_stabilizers} provided below relies on Proposition~\ref{prop:stabilizer_component}.

\begin{proof}
Any connected stabilizer is supported on a connected subgraph of the spacetime graph. As a result, any connected stabilizer with weight $\leq M$ is included in a ball with radius $\lfloor M/2 \rfloor$ of the spacetime. This guarantees that it will be discovered by Algorithm~\ref{algo:spacetime_code_local_genertors}.
\end{proof}

\section{Conclusion}
\label{sec:conclusion}

We proposed an efficient and versatile strategy for the correction of circuit faults in Clifford circuits by reducing to error correction of a stabilizer code.
The main advantage of our approach is its flexibility. It applies to any Clifford syndrome extraction circuit, including those of topological codes and Floquet codes, and eliminates the tedious step of mapping circuit outcomes to corresponding syndromes. It also applies to general Clifford circuits which are not necessarily based on a quantum code.

Our scheme can be used to automatically generate low-weight checks in Clifford circuits.
Alternatively, it may be useful as a compilation tool for detecting and removing redundancy in a Clifford circuit, thereby reducing circuit size.
Adapting Algorithm~\ref{algo:circuit_checks} for this task is immediate.

In the future, our scheme may be improved by designing a decoder that exploits the equivalence between different fault configurations like in the work of Pryadko~\cite{pryadko2020maximum}.
Design of better decoders for quantum LDPC codes may also improve our scheme.

Conceptually, our work emphasizes and formally captures a circuit-centric approach of quantum error correction and fault tolerance.  This circuit-centric approach is central to Floquet codes~\cite{hastings2021dynamically} and has also lead to new ideas for surface codes~\cite{mcewen2023relaxing}.
Our formalism could be used to take this approach further by, for example, searching over the space of quantum codes and circuits.  In particular, machine learning techniques which were limited to Kitaev's codes in~\cite{nautrup2019optimizing} could be expanded to a much broader range of codes and circuits.

\section*{Acknowledgements}

We would like to thank David Aasen, Michael Beverland, Jeongwan Haah, Vadym Kliuchnikov and Marcus Silva for insightful discussions.
We thank Rui Chao for his comments.

%\bibliographystyle{plain}
%\bibliography{references}

\appendix

\section{Outcomes of the check operators}

This section relates the value of a check of the outcome code to the measurement outcome of a check operator, justifying the definition of check operators.
The first lemma provides a description of the outcomes flipped by a set of faults as a commutator of the corresponding fault operator.

\begin{lemma} [Outcome flip] \label{lemma:flip_from_propagation}
Let $F$ be a fault operator.
The faults corresponding to $F$ induce a flip of the measurement of $S_j$ iff $[\propagation F, \eta_{\ell_j - 0.5}(S_j)] = 1$.
\end{lemma}

\begin{proof}
Based on Proposition~\ref{prop:fault_effects}, the faults represented by $F$ induce a flip of the outcome of $S_j$ iff $[\propagation F_{\ell_j - 0.5}, S_j] = 1$ and because $\eta_{\ell_j - 0.5}(S_j)$ is trivial except at level $\ell_j - 0.5$, we have $[\propagation F_{\ell_j - 0.5}, S_j] = [\propagation F, \eta_{\ell_j - 0.5}(S_j)]$.
\end{proof}

The following Lemma shows that measurement of a check operator returns the value of the corresponding check of the outcome code.

\begin{lemma} [Check operator outcome] \label{lemma:outcome_check_and_check_operator}
Let $F$ be a fault operator.
The faults corresponding to $F$ induce a flip of an outcome check $u \in \outcome(\circuit)^\perp$ iff
$[F, \backpropagation{F(u)}] = 1$.
\end{lemma}

\begin{proof}
By definition of $F(u)$, we have
\begin{align*}
[F, \backpropagation{F(u)}]
	& = [F, \backpropagation{\prod_{j = 1}^{m} \eta_{\ell_j - 0.5}(S_j^{u_j})}]
\end{align*}
which leads to
\begin{align*}
[F, \backpropagation{F(u)}]
	& = [F, \prod_{j = 1}^{m} \backpropagation{\eta_{\ell_j - 0.5}(S_j^{u_j})}]
	& \text{Using Eq.~\eqref{eq:backpropagation_product}}
\\
	& = \sum_{j = 1}^{m} [F, \backpropagation{\eta_{\ell_j - 0.5}(S_j^{u_j})}]
	& \text{Using Eq.~\eqref{eq:commutator_bilinearity}}
\\
	& = \sum_{j = 1}^{m} [\propagation{F}, \eta_{\ell_j - 0.5}(S_j^{u_j})]
	& \text{By Proposition.~\ref{prop:propagated_commutator}}
\end{align*}
By Lemma.~\ref{lemma:flip_from_propagation}, this last sum coincides with the check 
$
\sum_{j} u_j o_j
$
of the outcome code corresponding to the vector $u$.
\end{proof}

The next lemma states that an error $S_j$ right before or right after the measurement of $S_j$ does not flip any of the check operator outcomes.

\begin{lemma} [Stabilizer error] \label{lemma:sabilizer_error_vs_relation_operator}
For all $u \in \outcome(\circuit)^\perp$ and for all measured operators $S_j$, we have
$
[{\eta_{\ell_j \pm 0.5}(S_j)}, \backpropagation{F(u)}] = 0.
$
\end{lemma}

\begin{proof}
Because $S_j$ belongs to the stabilizer group of the system right after level $\ell_j$, the faults corresponding to ${\eta_{\ell_j + 0.5}(S_j)}$ do not flip any of the checks of the outcome code. 
Based on Lemma~\ref{lemma:outcome_check_and_check_operator}, this leads to the lemma for the sign $\ell_j + 0.5$.
The same result holds for $\ell_j - 0.5$ because a fault $S_j$ after the measurement of $S_j$ is equivalent to a fault $S_j$ before this measurement.
\end{proof}

A Pauli error on the input state also keeps the check operator outcomes trivial.

\begin{lemma} [Input error] \label{lemma:input_error_vs_relation_operator}
For all $u \in \outcome(\circuit)^\perp$ and for all $P \in \Pauli_n$, we have
$
[\eta_{0.5}(P), \backpropagation{F(u)}] = 0.
$
\end{lemma}

\begin{proof}
An fault $P$ on the input state of the circuit cannot flip the value of a check because the outcome code $\outcome(\circuit)$ is the set of checks for all possible input states.
Based on Lemma~\ref{lemma:outcome_check_and_check_operator}, this implies $[\eta_{0.5}(P), \backpropagation{F(u)}] = 0$ for all check $u$ of the outcome code.
\end{proof}

\section{Commutation of the check operators}

The goal of this section is to prove that check operators are pairwise commuting.
More precisely, we show by induction on $\ell$ that $[\backpropagation{F(u)}_{\ell + 0.5}, \backpropagation{F(v)}_{\ell + 0.5}] = 0$ for all $\ell$.

\begin{prop} \label{prop:check_operators_commutation}
Let $\circuit$ be a Clifford circuit.
For all $u, v \in \outcome(\circuit)^\perp$, we have 
$
[\backpropagation{F(u)}, \backpropagation{F(v)}] = 0.
$
\end{prop}

\begin{proof}
Based on Eq.~\eqref{eq:commutator_tensor_product}, we have
\begin{align*}
[\backpropagation{F(u)}, \backpropagation{F(v)}]
= \sum_{\ell = 0}^{\Delta}
	[\backpropagation{F(u)}_{\ell + 0.5}, \backpropagation{F(v)}_{\ell + 0.5}]
\end{align*}
In the remainder of this proof, we demonstrate by induction on $\ell$ (in decreasing order) that 
$[\backpropagation{F(u)}_{\ell + 0.5}, \backpropagation{F(v)}_{\ell + 0.5}] = 0$ for all $\ell$.

For $\ell = \Delta$, by definition we have 
$
{F(u)_{\Delta + 0.5}} = {F(v)_{\Delta + 0.5}} = I
$
and therefore
$
\backpropagation{F(u)}_{\Delta + 0.5} = \backpropagation{F(v)}_{\Delta + 0.5} = I
$
which yields 
$
[\backpropagation{F(u)}_{\Delta + 0.5}, \backpropagation{F(v)}_{\Delta + 0.5}] = 0.
$

Assume now that $[\backpropagation{F(u)}_{\ell+1.5}, \backpropagation{F(v)}_{\ell+1.5}] = 0$ and let us prove that $[\backpropagation{F(u)}_{\ell + 0.5}, \backpropagation{F(v)}_{\ell + 0.5}] = 0$.
Using Proposition~\ref{prop:explicit_propagation} and Eq.~\eqref{eq:U_composition}, we find
\begin{align*}
\backpropagation{F(u)}_{\ell + 0.5} 
	& = \prod_{j = \ell}^{\Delta} U_{\ell, j}^{-1} F(u)_{j + 0.5} U_{\ell, j}
\\
	& = F(u)_{\ell + 0.5} \prod_{j = \ell+1}^{\Delta} U_{\ell, j}^{-1} F(u)_{j + 0.5} U_{\ell, j}
\\
	& = F(u)_{\ell + 0.5} \prod_{j = \ell+1}^{\Delta} 
		\left( U_{\ell+1, j} U_{\ell, \ell+1} \right)^{-1} 
		F(u)_{j + 0.5} 
		U_{\ell+1, j} U_{\ell, \ell+1}
\\
	& = F(u)_{\ell + 0.5} U_{\ell, \ell+1}^{-1} \left( \prod_{j = \ell+1}^{\Delta} 
		U_{\ell+1, j}^{-1} 
		F(u)_{j + 0.5} 
		U_{\ell+1, j} \right) U_{\ell, \ell+1}
\\
	& = F(u)_{\ell + 0.5} U_{\ell, \ell+1}^{-1} \backpropagation{F(u)}_{\ell+1.5} U_{\ell, \ell+1}
\end{align*}
The same holds for $v$, that is 
$\backpropagation{F(v)}_{\ell + 0.5} \
= 
F(v)_{\ell + 0.5} U_{\ell, \ell+1}^{-1} \backpropagation{F(v)}_{\ell+1.5} U_{\ell, \ell+1}.
$
Using these expressions for $\backpropagation{F(u)}_{\ell + 0.5}$ and $\backpropagation{F(v)}_{\ell + 0.5}$ and applying Eq.~\eqref{eq:commutator_bilinearity}, we obtain
\begin{align*}
[\backpropagation{F(u)}_{\ell + 0.5}, \backpropagation{F(v)}_{\ell + 0.5}]
%	& = 
%	[
%	F(u)_{\ell + 0.5} U_{\ell, \ell+1}^{-1} \backpropagation{F(u)}_{\ell+1 + 0.5} U_{\ell, \ell+1}
%	,
%	F(v)_{\ell + 0.5} U_{\ell, \ell+1}^{-1} \backpropagation{F(v)}_{\ell+1 + 0.5} U_{\ell, \ell+1}
%	]
%\\
	& = [F(u)_{\ell + 0.5}, F(v)_{\ell + 0.5}]
\\
	& \phantom{ = } + 
	[
		U_{\ell, \ell+1}^{-1} \backpropagation{F(u)}_{\ell+1 + 0.5} U_{\ell, \ell+1}
		,
		U_{\ell, \ell+1}^{-1} \backpropagation{F(v)}_{\ell+1 + 0.5} U_{\ell, \ell+1}
	] 
\\
	&  \phantom{=} + [F(u)_{\ell + 0.5}, U_{\ell, \ell+1}^{-1} \backpropagation{F(v)}_{\ell+1 + 0.5} U_{\ell, \ell+1}]
\\
	& \phantom{ = } + [U_{\ell, \ell+1}^{-1} \backpropagation{F(u)}_{\ell+1 + 0.5} U_{\ell, \ell+1}, F(v)_{\ell + 0.5}] \cdot
\end{align*}
Let us show that each of these four terms is trivial.
The first one 
$
[F(u)_{\ell + 0.5}, F(v)_{\ell + 0.5}]
$
is trivial by definition of the $F(u)$ and $F(v)$.
Using Eq.~\eqref{eq:commutator_unitary} and the induction hypothesis, we get 
$$
[
U_{\ell, \ell+1}^{-1} \backpropagation{F(u)}_{\ell+1.5} U_{\ell, \ell+1}
,
U_{\ell, \ell+1}^{-1} \backpropagation{F(v)}_{\ell+1.5} U_{\ell, \ell+1}
]
=
[\backpropagation{F(u)}_{\ell+1.5}, \backpropagation{F(v)}_{\ell+1.5}]
= 0
$$
proving that the second term is trivial.
Because the operator $U_{\ell, \ell+1}$ is the product of the unitary gates at level $\ell+1$ and $F(u)_{\ell + 0.5}$ and $F(v)_{\ell + 0.5}$ corresponds to measurements at the same level, the support of $U_{\ell, \ell+1}$ cannot overlap with the supports of $F(u)_{\ell + 0.5}$ and $F(v)_{\ell + 0.5}$. 
As a result $U_{\ell, \ell+1}$ acts trivially on $F(u)_{\ell + 0.5}$ and $F(v)_{\ell + 0.5}$.
Therefore, we have 
\begin{align*}
[F(u)_{\ell + 0.5}, U_{\ell, \ell+1}^{-1} \backpropagation{F(v)}_{\ell+1 + 0.5} U_{\ell, \ell+1}]
	& = [U_{\ell, \ell+1} F(u)_{\ell + 0.5} U_{\ell, \ell+1}^{-1}, \backpropagation{F(v)}_{\ell+1.5}]
\\
	& = [F(u)_{\ell + 0.5}, \backpropagation{F(v)}_{\ell+1.5}] \cdot
\end{align*}
Therein, we used Eq.~\eqref{eq:commutator_unitary} in the first equality.
To see that $[F(u)_{\ell + 0.5}, \backpropagation{F(v)}_{\ell+1.5}] = 0$, write this commutator as
\begin{align*}
[F(u)_{\ell + 0.5}, \backpropagation{F(v)}_{\ell+1.5}]
	& = [\eta_{\ell+1.5}(F(u)_{\ell + 0.5}), \backpropagation{F(v)}] \cdot
\end{align*}
and apply Lemma~\ref{lemma:sabilizer_error_vs_relation_operator}.
We can apply this lemma because 
$
F(u)_{\ell + 0.5}
$
is a product of some measured operators $S_j$ with level $\ell_j = \ell + 1$ in the circuit.
We proved that the third term is trivial. 
By symmetry, the fourth term is trivial by the same argument.
This proves that 
$
[\backpropagation{F(u)}_{\ell + 0.5}, \backpropagation{F(v)}_{\ell + 0.5}] = 0
$
concluding the proof of the proposition.
\end{proof}

\section{Logical operators of the spacetime code}

To describe the logical operators of the spacetime code, we introduce some notation.
For any $\ell = 1, \dots, \Delta$ and for any $P \in \Pauli_n$, define the fault operator $G(P, \ell)$ as
\begin{align} \label{eq:image-operator}
G(P, \ell) = \eta_{\ell - 0.5}(P) \eta_{\ell + 0.5}(U_{\ell-1, \ell} P U_{\ell-1, \ell}^{-1}) \cdot
\end{align}
These operators satisfy
\begin{align} \label{eq:propagation_of_G}
\propagation{G(P, \ell)} = \eta_{\ell - 0.5}(P) \cdot
\end{align}
For any vector $v \in \Z_2^m$, define the fault operator $L(v)$ as 
\begin{align}
L(v) = \prod_{j=1}^m G(P_j, \ell_j)^{v_j}
\end{align}
where $\ell_j$ is the level of the $j$ th measured operator $S_j$ and $P_j \in \Pauli_n$ is an arbitrary Pauli operator acting on the support of $S_j$ that anti-commutes with $S_j$.
Because $P_j$ is included in the support of $S_j$ it commutes with all other measured operators at level $\ell_j$.

\begin{prop} \label{prop:logical_operators}
Let $\circuit$ be a depth-$\Delta$ Clifford circuit acting on $n$ qubits.
The stabilizers and logical operators of the spacetime code are generated by the following operators.
\begin{itemize}
    \item A set of operators $\eta_{\Delta + 0.5}(P)$ where $P \in \Pauli_n$ runs over a basis of $\Pauli_n$.
    \item A set of operators $G(P, \ell)$ for all $\ell = 1, \dots, \Delta$ where $P \in \Pauli_n$ runs over a basis of the space of Pauli operators that commutes with all the measured operators at level $\ell$.
    \item A set of operators $L(v)$ where $v \in \Z_2^m$ runs over a basis of the space $\outcome(\circuit)$.
\end{itemize}
\end{prop}

The following lemma is used in the proof of the proposition.

\begin{lemma} \label{lemma:relation_Fu_Lv}
For all $u, v \in \Z_2^m$, we have
$
[\backpropagation{F(u)}, L(v)] = (u | v).
$
\end{lemma}

\begin{proof}
\begin{align*}
[\backpropagation{F(u)}, L(v)]
    & = [F(u), \propagation{L(v)}]
    & \text{By Proposition~\ref{prop:propagated_commutator}}
\\
    & = [F(u), \propagation{\prod_{j=1}^m L(P_j, \ell_j)^{v_j}}]
    & \text{By definition of $L(v)$}
\\
    & = [F(u), \prod_{j=1}^m \propagation{G(P_j, \ell_j)^{v_j}}]
    & \text{By Eq.~\eqref{eq:propagation_product}}
\\
    & = \sum_{j=1}^m v_j [F(u), \propagation{G(P_j, \ell_j)}]
    & \text{By Eq.~\eqref{eq:commutator_bilinearity}}
\\
    & = \sum_{j=1}^m v_j [F(u), \eta_{\ell_j - 0.5}(P_j)]
    & \text{By Eq.~\eqref{eq:propagation_of_G}}
\\
    & = \sum_{j=1}^m v_j [F(u)_{\ell_j - 0.5}, P_j]
    &
\end{align*}
By definition of $P_i$, we have $[F(u)_{\ell_j - 0.5}, P_j] = u_j$ which yields
$
[\backpropagation{F(u)}, L(v)] = (u | v).
$
\end{proof}

\begin{proof} [Proof of Proposition~\ref{prop:logical_operators}]
First, let us prove that these three families of operators are logical operators of the spacetime code, that is that they commute with all stabilizers $\backpropagation{F(u)}$ with $u \in \outcome(\circuit)^\perp$.

The operators of the form $\eta_{\Delta + 0.5}(P)$ satisfy
\begin{align*}
[\backpropagation{F(u)}, \eta_{\Delta + 0.5}(P)]
    & = [F(u), \propagation{\eta_{\Delta + 0.5}(P)}]
\\
    & = [F(u), \eta_{\Delta + 0.5}(P)]
% \\
%     & = [P, F(u)_{\Delta + 0.5}]
% \\
%     & = [P, I]
%     = 0 \cdot
\end{align*}
which is trivial because $F(u)$ is trivial over level $\Delta + 0.5$.

For the operators $G(P, \ell)$, we get
$$
[\backpropagation{F(u)}, G(P, \ell)]
= 
[F(u), \propagation{G(P, \ell)}]
$$
which is equal to $[F(u), \eta_{\ell - 0.5}(P)]$ by Eq.~\ref{eq:propagation_of_G}.
Because $P$ commute with all the measured operators at level $\ell$, we get
\begin{align*}
[\backpropagation{F(u)}, G(P, \ell)]
    & = [\eta_{\ell - 0.5}(P), F(u)]
\\
    & = [F(u)_{\ell - 0.5}, P] = 0 \cdot
\end{align*}

For the operators $L(v)$, based on Lemma~\ref{lemma:relation_Fu_Lv} we have
$
[\backpropagation{F(u)}, L(v)]
= (u | v)
$
and this inner product is trivial because $v$ is in the code $\outcome(\circuit)$ and $u$ belongs to its dual.

To prove that these three families of operators generate all stabilizers and logical operators, it is enough to show that the group they generate has rank $2K + R$ where $K$ is the number of logical qubits of the stabilizer code and $R$ is the rank of the stabilizer group.
For the spacetime code, we know that $K = n(\Delta+1) - r$ and $R = r$

Denote by $L_1, L_2, L_3$ the subgroups of $\Pauli_n(\delta+1)$ generated by these three sets of operators.
We have $L_1 \cap L_2 = \{I\}$ because the operators of $L_2$ are supported on at least two levels and $L_1$ has support on level $\Delta + 0.5$.
We also have $L_1 \cap L_3 = \{I\}$ because the operators of $L_1$ cannot flip any outcome and the only operator of $L_3$ that induces no outcome flip is $I$.
The same argument also shows that $L_2 \cap L_3 = \{I\}$.
As a result, the rank of the subgroup of $\Pauli_n(\delta+1)$ generated by all the operators of $L_1, L_2$ and $L_3$ is the sum of the ranks of the three subgroups $\rank(L_1) + \rank(L_2) + \rank(L_3)$.

It is immediate to see that $\rank(L_1) = 2n$.
For a circuit without measurement the rank of $L_2$ is $2n\Delta$.
The constraint associated with the commutation with each measurement decreases the rank by $1$, which yields $\rank(L_2) = 2n\Delta - m$ where $m$ is the number of measurements of the circuit.
Finally, the rank of $L_3$ is given by the dimension of the outcome code, that is $\rank(L_3) = m-r$.

Putting things together this proves that these three sets of operators generate a group with rank 
$
% 2n + 2n\Delta - m + m - r 
% =
2n(\Delta+1) - r
$
which coincides with the value of $2K + R$.
This proves that this family of operators generate all stabilizer and logical operators.
\end{proof}

\section{Levels supporting a check operator}

In this section, we prove the support of a check operator of the spacetime code is related to the support of the measurements involved in the corresponding check of the outcome code.

\begin{lemma} \label{lemma:check_operator_support}
Let $u \in \outcome(\circuit)^\perp$ be a non-zero vector. Denote by $\ell_u$ (respectively $\ell_u'$) the minimum (respectively maximum) level of a measured operator $S_j$ with $u_j = 1$.
If $\ell < \ell_u$ or $\ell \geq \ell_u'$ then we have
$
\backpropagation{F(u)}_{\ell + 0.5} = I.
$
\end{lemma}

\begin{proof}
By definition of $\ell_u'$, for all $\ell \geq \ell_u'$ the component $F(u)_{\ell + 0.5}$ is trivial.
Applying Proposition~\ref{prop:explicit_propagation}, this shows that $\backpropagation{F(u)}_{\ell + 0.5}$ is also trivial for all $\ell \geq \ell_u'$.

Assume now that $\ell = 0$.
We know from Lemma~\ref{lemma:input_error_vs_relation_operator} that $\backpropagation{F(u)}$ commutes with all fault operators $\eta_{0.5}(P)$ supported on level $0.5$.
Because $P$ can be any operator of $\Pauli_n$, this implies $\backpropagation{F(u)}_{0.5} = I$.

Consider now the case $\ell < \ell_u$ and let us prove that 
$
\backpropagation{F(u)}_{\ell + 0.5} = \backpropagation{F(u)}_{0.5}
$
which is the identity.
Applying Proposition~\ref{prop:explicit_propagation}, we find
\begin{align*}
\backpropagation{F(u)}_{0.5} 
%= \prod_{j = 0}^{\Delta} \left( U_{0, j}^{-1} F_{j + 0.5} U_{0, j} \right)
= \prod_{j = \ell}^{\Delta} \left( U_{0, j}^{-1} F_{j + 0.5} U_{0, j} \right)
\end{align*}
because for all $j < \ell$, we know that $F(u)_{j + 0.5} = I$.
Using Eq.~\eqref{eq:U_composition}, we can write 
$
U_{0, j} = U_{\ell, j} U_{0, \ell}
$
for all $j \leq \ell$ which yields
\begin{align*}
\backpropagation{F(u)}_{0.5} 
	& = \prod_{j = \ell}^{\Delta} \left((U_{\ell, j} U_{0, \ell})^{-1} F_{j + 0.5} U_{\ell, j} U_{0, \ell} \right)
\\
	& = U_{0, \ell}^{-1} \prod_{j = \ell}^{\Delta} \left( U_{\ell, j}^{-1} F_{j + 0.5} U_{\ell, j} \right) U_{0, \ell}
\\
	& = U_{0, \ell}^{-1} \backpropagation{F(u)}_{\ell + 0.5} U_{0, \ell} \cdot
\end{align*}
Because $\backpropagation{F(u)}_{0.5}$ is trivial, this relation implies $\backpropagation{F(u)}_{\ell + 0.5} = U_{0, \ell} \backpropagation{F(u)}_{0.5} U_{0, \ell}^{-1} = I$.
\end{proof}

\section{Alternative definition of the check operators}
\label{appendix:alternative_def_of_check_operators}

One may prefer to use the cumulant obtained by propagating faults forward instead of the back-cumulant obtained from the backward propagation of faults.
Here, we show that we can obtain the check operators by propagating forward the fault operators ${F'(u)}$ defined in Section~\ref{subsec:check_operators}.

\begin{prop} \label{prop:fault_operator_alternative_def}
If $u \in \outcome(\circuit)^\perp$, then, we have
\begin{align}
\backpropagation{F(u)} = \propagation{F'(u)} \cdot
\end{align}
\end{prop}

Let us first prove a lemma.

\begin{lemma} \label{lemma:relation_F_F'}
Let $u \in \Z_2^m$.
If $\backpropagation{F(u)}_{0.5} = I$ then we have
\begin{align}
\backpropagation{F(u)} = \propagation{F'(u)} \cdot
\end{align}
\end{lemma}

\begin{proof}
Let us prove by induction that we have the equality 
$
\backpropagation{F(u)}_{\ell + 0.5} = \propagation{F'(u)}_{\ell + 0.5}
$
for all level $\ell$.

For $\ell = 0$, we have $\backpropagation{F(u)}_{0.5} = I$ by assumption and by definition of the cumulant we have $\propagation{F'(u)}_{0.5} = I$.

Assume that the result is true for level $\ell - 1 < \Delta$.
Using Proposition~\ref{prop:explicit_propagation}, we obtain
\begin{align} \label{eq:proof_lemma_F_F'_prop}
\propagation{F'(u)}_{\ell + 0.5} 
    & = U_{\ell} \propagation{F'(u)}_{\ell - 0.5} U_\ell^{-1} F'(u)_{\ell + 0.5}
\end{align}
and
\begin{align} \label{eq:proof_lemma_F_F'_backprop}
\backpropagation{F(u)}_{\ell - 0.5}
    & = U_{\ell}^{-1} \backpropagation{F(u)}_{\ell + 0.5} U_{\ell} F(u)_{\ell - 0.5} \cdot
\end{align}
Using the induction hypothesis, we can replace $\propagation{F'(u)}_{\ell - 0.5}$ by $\backpropagation{F(u)}_{\ell - 0.5}$ in Eq.~\eqref{eq:proof_lemma_F_F'_prop} and we can use Eq.~\eqref{eq:proof_lemma_F_F'_backprop} for the value of $\backpropagation{F(u)}_{\ell - 0.5}$.
This leads to the equation
\begin{align*}
\propagation{F'(u)}_{\ell + 0.5} 
    & = U_{\ell} \backpropagation{F(u)}_{\ell - 0.5} U_\ell^{-1} F'(u)_{\ell + 0.5}
\\
    & = U_{\ell} U_{\ell}^{-1} \backpropagation{F(u)}_{\ell + 0.5} U_{\ell} F(u)_{\ell - 0.5} U_\ell^{-1} F'(u)_{\ell + 0.5}
\\
    & = \backpropagation{F(u)}_{\ell + 0.5} U_{\ell} F(u)_{\ell - 0.5} U_\ell^{-1} F'(u)_{\ell + 0.5}
\end{align*}
By definition of $U_\ell$ acts trivially on $F(u)_{\ell - 0.5}$ implying 
$
U_{\ell} F(u)_{\ell - 0.5} U_{\ell}^{-1} = F(u)_{\ell - 0.5}.
$
Injecting this in the previous equation produces
\begin{align*}
\propagation{F'(u)}_{\ell + 0.5} 
    & = \backpropagation{F(u)}_{\ell + 0.5} F(u)_{\ell - 0.5} F'(u)_{\ell + 0.5}
\end{align*}
which leads to
$
\propagation{F'(u)}_{\ell + 0.5}
=
\backpropagation{F(u)}_{\ell + 0.5}
$
because $F(u)_{\ell - 0.5} = F'(u)_{\ell + 0.5}$ by definition.
\end{proof}

\begin{proof} [Proof of Proposition~\ref{prop:fault_operator_alternative_def}]
By Lemma~\ref{lemma:input_error_vs_relation_operator}, we know that $\backpropagation{F(u)}_{0.5}$ commutes with all Pauli operators $P \in \Pauli_n$.
The only Pauli operator that satisfies this property is $\backpropagation{F(u)}_{0.5} = I$.
Therefore, the operator $\backpropagation{F(u)}$ satisfies the assumption of Lemma~\ref{lemma:relation_F_F'} which proves the Proposition.
\end{proof}

\section{The connected components of a stabilizer of the spacetime code}
\label{sec:connected-components}

The proof of proposition~\ref{prop:algo_return_all_connected_stabilizers} relies on the following proposition. The proof of this proposition relies on lemmas proven after the proposition.

\begin{prop} \label{prop:stabilizer_component}
Let $S$ be a stabilizer for a spacetime code.
The restriction $S_{|\kappa}$ of $S$ to a connected component $\kappa$ of the support of $S$ in the spacetime graph is a stabilizer of the spacetime code.
\end{prop}

\begin{proof}
Based on lemma~\ref{lemma:space_time_stabilizer_decomposition}, the restriction $S_{|\kappa}$ can be written as $\backpropagation{F(v)}$ for some vector $v$. To show that this is a stabilizer of the spacetime code, we must show that $v \in \outcome(\circuit)^\perp$.
% It suffices to show that there is a relation between the measured operators $S_{j}$ such that $v_j = 1$.
Denote by $\ell_v'$ the maximum level of a measured operator $S_{j}$ of the circuit such that $v_j = 1$.
The product of the measured operators indexed by $v$ at level $\ell_v'$ is $F(v)_{\ell_v' - 0.5}$.
By Lemma~\ref{lemma:last_level_F_expression}, this operator is equal to $\propagation{F'(v)}_{\ell_v' - 0.5}$.
This is a stabilizer for the output of the circuit $\circuit[\ell_v'-1]$ based on Lemma~\ref{lemma:stabilizer_of_subcircuit} (recall that $\circuit[\ell_v']$ is the subcircuit containing all the operations of $\circuit$ with level $\leq \ell_v'$).
This implies that, in the absence of circuit faults, the outcomes of the measurements of the operators $S_j$ with $v_j = 1$ have a fixed parity 0, which means that $(o | v) = 0$ for all $o \in \outcome(\circuit)$. This proves the proposition.
\end{proof}

\begin{lemma} \label{lemma:space_time_stabilizer_decomposition}
Let $\circuit$ be a Clifford circuit with $m$ measurements.
Let $S$ be a stabilizer for the corresponding spacetime code. We can decompose $S$ as
$$
S = \prod_{i} \backpropagation{F(u(i))}
$$
where $u(i) \in \Z_2^m$.
Moreover, the support of the operators $\backpropagation{F(u(i))}$ have disjoint supports.
\end{lemma}

\begin{proof}
We can decompose $S$ as
$$
S = \prod_{i} S_{|\kappa_i}
$$
where $\kappa_1, \kappa_2, \dots$ are the connected components of the support of $S$ in the spacetime graph.
By construction, the operators $S_{|\kappa_i}$ have disjoint support.

Because it is a stabilizer, the operator $S$ is of the form $\backpropagation{F(u)}$ with $u \in \outcome(\circuit)^\perp$.
Given a measured operator $S_j$ of the circuit with level $\ell_j$, define $V(S_j)$ to be the set of vertices of spacetime graph of the form $(\ell_j \pm 0.5, q)$ where $q$ is a qubit of the support of $S_j$.
Define the vector $u(i) \in \Z_2^m$ to be the restriction of $u$ to the set of coordinates $j$ such that the circuit operation $V(S_j)$ overlaps with $\kappa_i$.
Because it is a clique of the spacetime graph, a set $V(S_j)$ can only overlap with a single connected component $\kappa_i$.
Therefore, $u$ can be decomposed as
$
u = \sum_{i} u(i)
$
and the vectors $u(i)$ do not overlap.
Applying Proposition~\ref{prop:check_operator_morphism}, this yields
$$
\backpropagation{F(u)} = \prod_{i} \backpropagation{F(u(i))} \cdot
$$

Now let us show that the support of $\backpropagation{F(u(i))}$ is the component $\kappa_i$.
By definition of the back-cumulant, if a qubit $Q = (\ell + 0.5, q)$ belongs to the support of $\backpropagation{F(u(i))}$, there must exists a path in this support that connects $Q$ to a set $V(S_j)$ for some $j \in \supp(u(i))$ (otherwise one cannot reach $Q$ by back-propagating $F(u(i))$). By definition of the set $u(i)$, the proves that the qubit $Q$ is part of the connected component $\kappa_i$.
This shows that $\supp(\backpropagation{F(u(i))}) = \kappa_i$ which implies $\backpropagation{F(u(i))} = S_{\kappa_i}$ and leads to the decomposition claimed in the lemma.
\end{proof}

\begin{lemma} \label{lemma:last_level_F_expression}
Let $\circuit$ be a Clifford circuit with depth $\Delta$ and with $m$ measurements.
Let $v \in \Z_2^m$.
If $\backpropagation{F(v)}_{0.5} = I$, then we have 
$$
F(v)_{\ell_v' - 0.5} = \propagation{F'(v)}_{\ell_v' - 0.5}
$$
where $\ell_v'$ is the maximum level of a measured operator $S_j$ with $v_j = 1$.
\end{lemma}

\begin{proof}
By definition, level $\ell_v' - 0.5$ is the largest non-trivial level of $\backpropagation{F(v)}$ and for this level we have
$$
\backpropagation{F(v)}_{\ell_v' - 0.5} = F(v)_{\ell_v' - 0.5}.
$$
Combining this with Lemma~\ref{lemma:relation_F_F'} this proves result.
\end{proof}

In the following lemma $\circuit[\ell]$ denotes the subcircuit of $\circuit$ that contains all the operations of $\circuit$ with level $\leq \ell$.

\begin{lemma} \label{lemma:stabilizer_of_subcircuit}
Let $\circuit$ be a Clifford circuit with depth $\Delta$ and with $m$ measurements.
Let $u \in \outcome(\circuit)^\perp$ and let $v \in \Z_2^m$ such that $\backpropagation{F(v)}$ is the restriction of $\backpropagation{F(u)}$ to one of a connected component of its support in the spacetime graph.
Then for all $\ell = 1, \dots \Delta$, the operator $\propagation{F'(v)}_{\ell + 0.5}$ belongs to output stabilizer group of the subcircuit $\circuit[\ell]$.
\end{lemma}

\begin{proof}
Let us prove this result by induction on $\ell$.

For $\ell = 1$, the result holds because the stabilizer group of the output of the subcircuit $\circuit[1]$ is generated by the measured operators at level $1$ and $\propagation{F'(v)}_{1.5} = F'(v)_{1.5}$ is a product of measured operators at level $1$.

Assume now that $\propagation{F'(v)}_{\ell - 1 + 0.5}$ is a stabilizer of the output state of the circuit $\circuit[\ell-1]$ and let us prove that the result is true at level $\ell$.

The measurements performed at level $\ell$ do not affect the stabilizer $\propagation{F'(v)}_{\ell - 1 + 0.5}$ of the circuit because it commutes with these measurements.
Indeed, consider a measured operator $M$ of the level $\ell$.
From Lemma~\ref{lemma:sabilizer_error_vs_relation_operator}, we know that $[\eta_{\ell - 0.5}(M), \backpropagation{F(u)}] = 0$ because
$\backpropagation{F(u)}$ is a stabilizer of the spacetime code.
This equation leads to 
$
[M, \backpropagation{F(u)}_{\ell - 0.5}] = 0
$
and
$
[M, \propagation{F'(u)}_{\ell - 0.5}] = 0
$
using Lemma~\ref{lemma:relation_F_F'}.
By definition of the spacetime graph of the circuit, the operation $M$ can only overlap with at most one of the connected component of the $\propagation{F'(u)}$.
Therefore, we also have 
$
[M, \propagation{F'(v)}_{\ell - 0.5}] = 0.
$
This proves that the stabilizer $\propagation{F'(v)}_{\ell - 0.5}$ commutes with all the operators measured at level $\ell$.
As a result, it is still a stabilizer after applying these measurements.
Moreover, the measured operators are also stabilizers.

Consider now the effect of the unitary operations.
The unitary operations of level $\ell$ induce a conjugation of the stabilizers by $U_\ell$ which is the product of all unitary operations at level $\ell$.
This maps the stabilizer $\propagation{F'(v)}_{\ell - 1 + 0.5}$ onto
\begin{align} \label{proof:stabilizer_update}
U_\ell \propagation{F'(v)}_{\ell - 1 + 0.5} U_\ell
%     & =
%     \prod_{j = 0}^{\ell-1} U_{\ell} U_{j,\ell-1} F'(v)_{j + 0.5} U_{j,\ell-1}^{-1} U_{\ell}^{-1}
% \\
    & = 
    \prod_{j = 0}^{\ell-1} U_{j,\ell} F'(v)_{j + 0.5} U_{j,\ell}^{-1} \cdot
\end{align}
Moreover, we know that $F'(v)_{\ell + 0.5}$ is a stabilizer because it is a product of measured operators at level $\ell$.
Multiplying the stabilizer obtained Eq.~\eqref{proof:stabilizer_update} with $F'(v)_{\ell + 0.5}$ we get
the stabilizer $\propagation{F'(v)}_{\ell + 0.5}$ for the output of the circuit $\circuit[\ell]$.
\end{proof}

\end{document}